\documentclass[12pt]{IEEEtran}
\usepackage{setspace} 
\IEEEoverridecommandlockouts
\onecolumn
\usepackage{cite}
\usepackage{amsmath,amssymb,amsfonts,amsthm,mathtools}
\usepackage{algorithm,algorithmic}
\usepackage[caption=false,font=footnotesize]{subfig}
\usepackage{pgfplots}
\usepackage{textcomp}
\usepackage{hyperref}
\usepackage{soul}
\pgfplotsset{compat=newest}
\newtheorem{theorem}{Theorem}

\newcommand{\vc}[1]{\mathbf{#1}}

\newcommand{\hb}{\vc{h}}
\newcommand{\gb}{\vc{g}}
\newcommand{\ab}{\vc{a}}

\newcommand{\Thetab}{\vc{\Theta}}
\newcommand{\wb}{\vc{w}}

\newcommand{\xb}{\vc{x}}

\newcommand{\yb}{\vc{y}}

\newcommand{\Fb}{\vc{F}}

\allowdisplaybreaks

\begin{document}
\doublespacing  
\title{Secure Beamforming in Multi-User Multi-IRS Millimeter Wave Systems}
\author{Anahid Rafieifar, Hosein Ahmadinejad, S. Mohammad Razavizadeh, and Jiguang He  
	\thanks{Anahid Rafieifar, Hosein Ahmadinejad, and S. Mohammad Razavizadeh are with School of Electrical Engineering, Iran University of Science and Technology (IUST) (emails:\{\href{anahid{\_}rafieifar@alumni.iust.fac.ir}{anahid{\_}rafieifar@alumni.iust.ac.ir},
		\href{ho\_ahmadinejad@alumni.iust.ac.ir}{ho\_ahmadinejad@alumni.iust.ac.ir}, and \href{smrazavi@iust.ac.ir}{smrazavi@iust.ac.ir}\}).}
	\thanks{Jiguang He is with the Technology Innovation Institute, 9639 Masdar city, Abu Dhabi, United Arab Emirates,
		and he is also with Centre for Wireless Communications,
		University of Oulu, 90014 Oulu, Finland (email:\{\href{jiguang.he@tii.ae}{jiguang.he@tii.ae}\}).}
	\thanks{\copyright~2023 IEEE. Personal use of this material is permitted. However, permission to use this material for any other purposes must be obtained from the IEEE by sending a request to pubs-permissions@ieee.org.}
}

\maketitle

\begin{abstract}
We study the secrecy rate maximization problem in a millimeter wave (mmWave) network, consisting of a base station (BS), multiple intelligent reflecting surfaces (IRSs) (or reconfigurable intelligent surfaces (RISs)), multiple users, and a single eavesdropper. To ensure a fair secrecy rate among all the users, we adopt a max-min fairness criterion which results in a mixed integer problem. We first relax discrete IRSs phase shifts to the continuous ones. To cope with the non-convexity of the relaxed optimization problem, we leverage the penalty method and block coordinate descent approach to divide it into two sub-problems, which are solved by successive convex approximation (SCA). Then, we propose a low-complexity mapping algorithm where feasible IRSs phase shifts are obtained. Mathematical evaluation shows the convergence of sub-problems to a Karush-Kuhn-Tucker (KKT) point of the original ones. Furthermore, the convergence guarantee of the overall proposed algorithm and computational complexity are investigated. Finally, simulation results show our proposed algorithm outweighs the conventional solutions based on the semi-definite programming (SDP) in terms of convergence and secrecy rate, especially in a larger number of IRSs and phase shifts where SDP suffers from rank-one approximation. Maximum ratio transmission (MRT) and IRS-free systems are also considered as other benchmarks.
\end{abstract}


\section{Introduction}
 Tremendous available bandwidth, low inter- and intra-cell interference, and small form factor of base station's (BS) antennas make millimeter wave (mmWave) as one of the most promising technologies for the fifth-generation (5G) and beyond wireless networks. However, small coverage, vulnerability to blockage, and large attenuation are issues bringing challenges to deploy such networks in practice. To tackle the aforementioned issues, increasing the number of BS's and users' antennas as well as employing relays have been proposed in literature \cite{mmWave1, mmWave2}. Intelligent reflecting surfaces (IRSs) have been recently introduced to reduce the high power consumption of antenna elements in relays and BSs \cite{IRS1, IRS2}. Besides, IRSs can enhance coverage, energy and spectral efficiency, and link quality in wireless communication systems. Commonly, IRSs are meta-surfaces with a controller and a large number of cost-effective and passive reflecting elements that can modify the signal propagation environment through changing the amplitude and phase shift of the impinging signals. Another attractive technology with a similar application to IRSs is holographic multiple-input multiple-output (MIMO) surfaces (HMIMOSs)\cite{jadid10}. Actually, HMIMOSs can serve as a transmitter, receiver,
 	or reflector, which can be considered as an extension of IRSs. Mainly, there are two categories of HMIMOSs based on the power consumption, namely, active HMIMOSs and passive ones. When HMIMOSs act as a transceiver, the active mode is adopted \cite{jadid11} and passive HMIMOSs are also supposed to be IRSs or reconfigurable intelligent surface (RISs) \cite{jadid12}.

\subsection{Related Works}
As stated earlier, IRS-assisted communication is one of the major active areas in wireless networks that has been attracting much attention \cite{Anahidconf, AH3, AH4, AH7, jadid3, jadid4, jadid5, jadid6,jadid13, jadid14 }. The first important aspect of IRSs is the ability of promoting spectral and energy efficiency in wireless networks. For example, in \cite{Anahidconf}, the performance of a multi-IRS-aided non-orthogonal multiple access (NOMA) in a cell-free massive MIMO system was elaborated. The maximization of weighted sum-rate in an IRS-assisted femtocell multiple-input single-output (MISO) system was also investigated in \cite{AH3} and \cite{AH4}. Energy-efficient design for downlink transmission of an IRS-aided multi-user MISO system was also considered in \cite{AH7}. The authors in \cite{jadid3} formulated a sum rate maximization problem in an IRS-assisted downlink MISO network. They simplified the problem by using the zero-forcing (ZF) transmission scheme and then optimized the transmit power and the continuous phase shifts by block coordinate descent (BCD) algorithm and majorization-minimization (MM) method. The similar optimization problem in an almost the same system model with continuous values for IRS's phase shifts is also formulated in \cite{jadid4}. Due to practical considerations, IRSs are not capable of adopting continuous phase-shifts, those only can accept a set of discrete ones. As a result, the problem of optimizing IRSs phase shifts turns from continuous one into a high complex integer programming. To deal with, continuous IRSs phase shift combined with semi-definite relaxation is suggested. In \cite{jadid5}, the authors maximized the weighted sum-rate for an IRS-aided multi-user system by designing a two-timescale joint active and passive beamforming algorithm under the assumption of discrete phase shifts. Due to passive elements of IRSs as well as a large number of communication links in IRS-based networks, channel estimation brings a new challenge. Thus, some novel methods have been  introduced to tackle this difficulty. For instance, \cite{jadid6} proposed a channel estimation scheme based on the discrete Fourier transform. In \cite{jadid13}, Wei et al. proposed a method for channel estimation based on the parallel factor decomposition algorithm in order to unfold the resulting cascaded channel. In \cite{jadid14}, the authors developed an algorithm based on estimation error minimization.

The benefits that are got from bringing IRSs in micro-wave wireless networks in terms of spectral and energy efficiency provide the motivation to extend these systems to the mmWave counterpart. Some poor aspects of mmWave networks such as sensitivity to blockage, small coverage region, and large attenuation, make IRS perfectly matched with these networks \cite{AH2,AH5,AH6, AH1}. In \cite{AH2}, the problem of weighted sum-rate maximization by jointly optimizing active and passive beamformers in a multi-user IRS-based mmWave system was investigated. Besides, the authors in \cite{AH5} and \cite{AH6} proposed an IRS-based mmWave system to maximize downlink sum-rate and minimize users' uplink transmit powers, respectively. Due to severe path-loss and sensitivity to blockage of mmWave systems compared with micro-wave alternative, employing multiple IRSs is recommended to create a hands-on degree-of-freedom to combat these issues. For example in \cite{AH1}, the problem of joint active and passive beamforming in the downlink of a single-user MISO mmWave system is studied where two cases single-IRS and multi-IRSs are considered to maximize the received signal power.   

Apart from the valuable effects of IRSs on spectral and energy efficiency, they have opened a new research area using their own beamforming ability in improving physical layer security\cite{AH8, AH9, AH12, AH10, AH13, AH11, jadid7,jadid8}.
In \cite{AH8, AH9, AH12, AH10}, IRS-assisted wiretap communications in the presence of an eavesdropper were considered which aimed to improve the physical layer security by maximizing the secrecy rate. In \cite{AH12}, the authors also assumed both continuous and discrete IRS reflecting coefficients and solved the proposed problem from a different point of view using the path-following algorithm in an iterative manner. The authors in \cite{AH10} also rose to the challenging scenario where eavesdroppers' channels had better quality than the legitimate users. Furthermore, in \cite{AH13}, the transmit power minimization problem in a system composed of one BS, one IRS, a single-antenna user, and one eavesdropper, under the constraint of minimum secrecy rate and unit modulus of phase shifts was analyzed. Such an objective is achieved by jointly optimizing transmit power and passive beamforming. Utilizing artificial noise (AN) in IRS-assisted systems is also another issue that was studied in \cite{AH11,jadid7} to enhance the achievable secrecy rate in the presence of eavesdroppers.
Instead of using AN in \cite{AH11,jadid7}, in \cite{jadid8} the authors considered a friendly jammer to cooperate with the transmitter, fight against multiple eavesdroppers, and maximize the secrecy rate.

Advantages of IRS networks have opened a new research area in terms of secrecy rate maximization in mmWave networks \cite{AH14,AH15, jadid1, jadid2, Anahid, AH16, AH17}. It does not seem convenient at first owning to high directivity in mmWave systems. However, drastic blockage as well as devastating path-loss cause lower degree-of-freedom for the BSs to control the beamforming pattern. That is to say, it is less likely to combat against the eavesdropper that is close to its target because of the lower capability of BS to adjust the transmission in order to suppress the eavesdropper signal. Motivated by this issue, the authors in \cite{AH14, AH15, jadid1, jadid2} considered single-IRS-assisted wireless communication systems in the presence of an eavesdropper in a mmWave scenario.
 Then, BS's and IRSs' beamforming vectors were jointly optimized to increase the secrecy rate. Due to the lower number of communication links and simplifying the proposed optimization problems, these works considered only one IRS. For further simplification, some also investigated exclusively single user scenario. There are a limited number works \cite{Anahid, AH16, AH17} which elaborated on employing multiple IRSs in a secure scenario. In our previous work \cite{Anahid}, we maximized secrecy rate by jointly optimizing active and passive beamforming at the BS and multiple IRSs, respectively. In that work, the optimization problem relevant to just a single user were solved.  It was highlighted that, multi-IRS systems can be more helpful than single-IRS ones especially for mmWave communication with limited coverage. In \cite{AH16}, secrecy rate was maximized by optimizing multiple IRSs' phase shifts as well as their on-off states. An extension of this work can be also found in \cite{AH17}. In these works also some theoretic assumptions such as considering line-of-sight (LoS) channels among IRSs, user, and eavesdropper have been considered which are not fully applicable in practical systems. Also one common problem of this multi-IRS systems is to investigate just only a single user. Moreover, these works mainly employed semi-definite programming-based (SDP-based) solution which in practice will suffer from a large number of IRSs due to the indispensable approximations. Therefore, the existence of a general system model with practical assumptions, formulation and solutions in such sophisticated scenarios is missing. 

Apart form the system model, IRS-based beamforming inherently imposes optimization-relevant challenges such as non-convex nature of unitary modulus constraint, discrete values of phase shifts as well as a particular form of beamforming matrices. Mostly, passive beamforming is performed in conjunction with active beamforming or power allocation. The discretized values of phase shifts are also mainly relaxed to continuous ones in order to turn the utterly complex integer problem into a handy one \cite{jadid1}. BCD is mostly used to break the problems into multiple simpler sub-problems \cite{AH1,AH5,AH8,jadid6, AH10, jadid7, jadid8, jadid1, jadid2, Anahid}. On the other hand, semi-definite programming (SDP) is also the most the common approach to deal with non-convexity of each sub-problem \cite{jadid6, AH10, jadid7, jadid8, jadid1, jadid2, Anahid, AH16, AH17}. Unfortunately, as SDP problem is formulated using definition of auxiliary variables, it leads to a non-convex rank-one constraint.  Gaussian randomization and/or eigen-value decomposition (EVD) methods are usually proposed to offset this problem. It has been shown, they can lead to acceptable results in single-IRS system with a small number of optimization variables. However, in multi-IRS systems with a large number of IRS phase shifts, SDP-based solution suffers in terms of computational complexity and optimality of solutions. 

After a comprehensive literature review, we find that there are few works that employs IRSs in mmWave networks for secure communication. In the existing works, simplified deployment such as either a single-user or single-IRS scenario has also been considered. In addition, in the existing multi-IRS system, SDP-based solution is mainly used. These observations motivate the study on the secrecy rate of a multi-user multi-IRS mmWave system in this paper.  
\subsection{Contributions} 
As stated earlier, in the literature \cite{AH14,AH15,jadid1, jadid2, Anahid, AH16, AH17}, a combination of mmWave communication with IRSs was investigated to improve the network security. However, all these works suffer to some extent. For instance, in \cite{AH14,AH15,jadid1, jadid2} to simplify the problem formulation and solution only the performance of a single IRS was investigated, which provided the motivation for employing multiple IRSs in secure mmWave communication in \cite{Anahid, AH16, AH17}. Nevertheless, these works only took into consideration single-user scenarios and assumed the LoS channels among IRSs, user, and eavesdropper to make their problems easy to tackle. Apart from that, they used SDP-based solutions with relaxations  which for a large number of reflecting elements in multiple IRSs make the usage questionable. To deal with these issues, in this paper, we study the physical layer security in a multi-user multi-IRS mmWave network. Maximization of the minimum secrecy rate (e.g., max-min problem) by jointly optimizing active and passive beamforming is investigated. Our proposed solution relies on continuous relaxation, BCD approach, successive convex approximation (SCA), and penalizing methods which outperforms the alternatives including the SDP-based one in the literature. Our main contributions can be summarized as follows:
\begin{itemize}
	\item We employ multiple IRSs in a multi-user mmWave network to improve the secrecy rate where BS has low degree-of-freedom to battle against eavesdropper due to high sensitivity to blockage and severe path-loss. In our proposed system model, both BS and IRSs are beamforming-capable in the presence of an eavesdropper. In addition, we rise to a challenging scenario, where the eavesdropper benefits from better channel condition than the legitimate user.
	\item Contrary to existing works, we consider direct links between BS and users as well as BS and eavesdropper. In addition, we assume more realistic and general channel models where both BS-user and IRS-user channels have multiple paths, including one LoS path and multiple non-LoS paths.
	\item As opposed to the literature solution methods mainly based on SDP, we solve the optimization problem with a novel approach. In which, at first we relax the hands-on discrete phase shift constraint to the continuous one. Then, using convex-concave decomposition method combined with penalty function approach more tractable constraints are formed. After that, we use BCD method to break the relaxed problem into two sub-problems, namely, active beamforming and passive beamforming. Each sub-problem is solved using SCA. Finally, feasible optimization variables are derived through simple but effective mapping algorithm.
	\item The convergence behaviors of the proposed algorithms are investigated. In particular, convergence to a Karush-Kuhn-Tucker (KKT) point for each sub-problem is mathematically proven. Convergence of the overall proposed algorithm is also analytically shown. Furthermore, the complexity analyses of all the algorithms are conducted.
	\item Finally, simulation results show that the IRS-assisted systems with the proposed scheme outperform the conventional SDP-based solution in terms of complexity as well as secrecy rate. IRS-free and maximum ratio transmission (MRT) are also considered as other benchmarks.
\end{itemize}
\subsection{Organization}
Section \ref{sec2} introduces system and signal model. In Section \ref{sec3}, secrecy rate maximization problem is formulated. In Section \ref{sec4}, we propose an algorithm to solve the optimization problem. Section \ref{sec5} studies the complexity and convergence of the proposed algorithm. Finally, Section \ref{sec6} presents the simulation results followed by conclusions in Section \ref{sec7}.
\section{System, Channel, and Signal Model}
\label{sec2}
\subsection{System Model}
MmWave systems suffer from low degrees-of-freedom compared with a micro-wave alternative in terms of signal propagation control due to the high-sensitivity to blockage and large attenuation. This issue can cause many problems during the presence of an eavesdropper because of limitations of BS to reduce the information leakage in mmWave systems. In such systems, employing IRSs can highly help with combating the security issue. Motivated by this, in this paper, we consider a mmWave downlink network consisting of one BS, $L$ IRSs, and $K$ legitimate users in the presence of an eavesdropper, as shown in Fig. \ref{fig1}. The BS and the $l$-th IRS are equipped with $M$ antennas and $N_l$ reflecting elements, respectively. The IRSs' reflecting elements are capable of changing the phase of the received signals in order to steer them towards the users' direction. To avoid unnecessary optimization complexity as well as focus on investigating the effect of multiple IRSs on the system performance in a secure communication scenario, we assume that both users and eavesdropper are equipped with a single antenna.
   \begin{figure}[t]
	\centering
	\includegraphics[width=9cm, height=8cm]{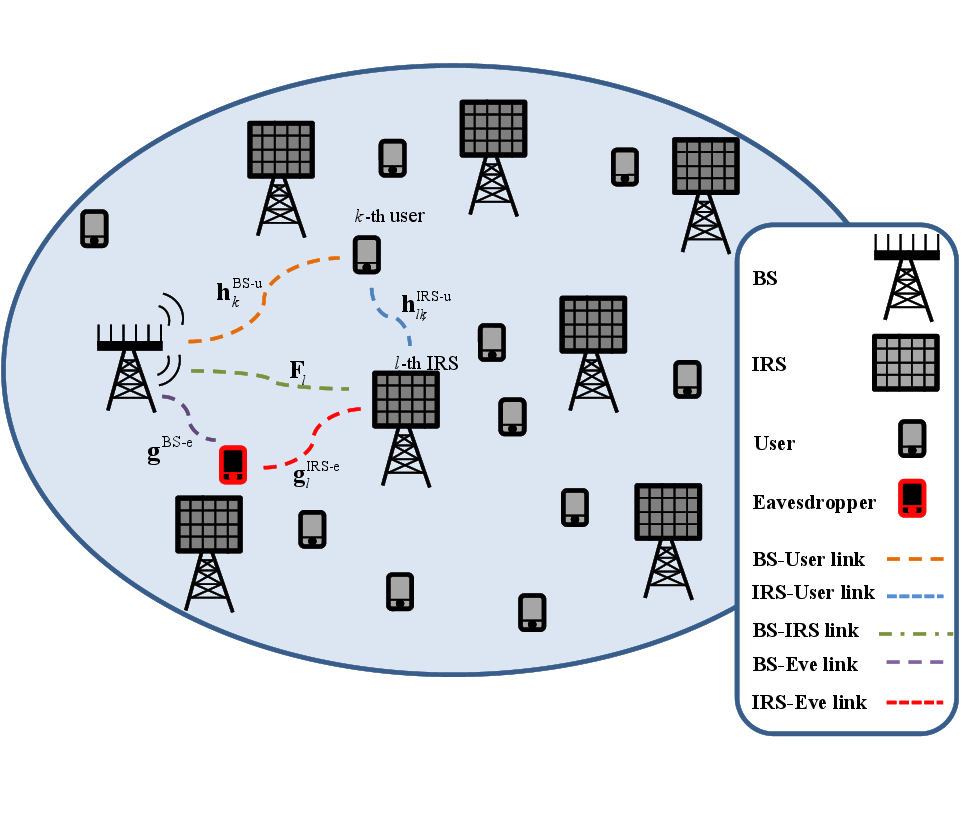}
	\caption{The secure communication system consisting of multiple users and IRSs, a single BS, and one eavesdropper.}
	\label{fig1}
\end{figure}	 
\subsection{Channel Model}
The channels from the BS to the $l$-th IRS, from the BS to the $k$-th user, from the BS to the eavesdropper, from the $l$-th IRS to the $k$-th user, and from the $l$-th IRS to the eavesdropper are denoted by 
   $\Fb_l\in\mathbb{C}^{N_l \times M}$ , $\hb_{k}^{\text{BS-u}}\in\mathbb{C}^{M \times 1}$, $\gb^{\text{BS-e}}\in\mathbb{C}^{M \times 1}$, $\hb_{l,k}^{\text{IRS-u}}\in\mathbb{C}^{N_l \times 1}$, and $\gb_{l}^{\text{IRS-e}}\in\mathbb{C}^{N_l \times 1}$, respectively. It is assumed that eavesdropper is connected to the BS and pretends to be a legitimate user and BS is also supposed to know the existence of an un-authorized user using some collected side information. In situations when the eavesdropper is not directly connected to the network, some methods such as CSI feedback or employing power leakage of eavesdroppers can be employed to acquire the eavesdropper's CSI \cite{EVE_CSI1, EVE_CSI2}. Thus, eavesdropper's channel coefficients are assumed to be known at the BS-side just as the legitimate users. Here, we use channel models similar to those in \cite{Anahid} and also make the assumption that the channels follow the block fading model. Furthermore, it is assumed that all channels' coefficients are known and perfectly estimated at the BS-side using the proposed channel estimation methods in \cite{jadid6, jadid13,jadid14,CSI1, CSI2, CSI3, CSI4}. The effect of Doppler frequency shift is also assumed to be completely compensated at the user-side. 
   \\
   The channel from the BS to the $k$-th user $\hb_k^{\text{BS-u}}$ is modeled according to the geometric channel model of the mmWave communications \cite{channel1} and can be expressed as 
   \begin{align}
   \label{eq1_channel}
   	\hb_k^{\text{BS-u}} = \frac{\sqrt M}{B} \sum_{\beta=1}^{B} \rho_{\beta,b}^k G_b  \ab_b(\phi_{\beta,b}^k),
   \end{align}
   where $B$ is total number of paths between BS and the $k$-th user, $G_b$ is the BS antenna gain, $\ab_b(\phi_{\beta,b}^k)$ is the normalized BS array response vector at azimuth angle of departure $\phi_{\beta,b}^k \in [0,2\pi]$, and $\rho_{\beta,b}^k \sim \mathcal{CN}(0,10^{-0.1\text{PL}_1 (d_{1,k})})$ is the complex gain of the $\beta$-th path between the BS and the $k$-th user. The path-loss can be obtained as \cite{channel2}
      \begin{align}
      \label{eq2_channel}
   	\text{PL}_1 (d_{1,k}) \: [\text{dB}] = \mu_1 + 10 \kappa_1 \: \log_{10} (d_{1,k}) + \xi_1,
   \end{align}
   where $d_{1,k}$, $\mu_1$, and $\kappa_1$ denote the distance between the BS and the $k$-th user, constant path-loss term, and path-loss exponent, respectively. In addition, $ \xi_1 \sim \mathcal{N}(0,\sigma_{\xi_1}^2) $ where $\sigma_{\xi_1}^2$ is shadowing variance. 
   
   Similar to \eqref{eq1_channel}, the BS-to-eavesdropper channel can be modeled as 	
   \begin{align}
   \label{eq3_channel}
   	\gb^{\text{BS-e}} = \frac{\sqrt M}{B} \sum_{\beta=1}^{B} \rho_{\beta,b}^e G_b \ab_b(\phi_{\beta,b}^e),
   \end{align}
   where $\rho_{\beta,b}^e \sim \mathcal{CN}(0,10^{-0.1\text{PL}_1 (d_{1,e})})$ is the complex gain of the $k$-th path between the BS and eavesdropper. Here, $d_{1,e}$ is the distance between the BS and eavesdropper. $\ab_b(\phi_{\beta,b}^e)$ is the normalized BS array response vector at azimuth angle of departure  $\phi_{\beta,b}^e \in [0,2\pi]$. 
   
   Assuming that the BS and IRSs are located at high altitudes, the channels between them are LoS dominant. Thus, the BS-to-the $l$-th IRS channel can be modeled as a rank-one matrix \cite{AH5}
   \begin{align}
   \label{eq4_channel}
   	\Fb_l = \sqrt{MN_l} \rho_b^l G_b \ab_l(\phi_l^b,\psi_l^b) \ab_b^H(\phi_b^l), \: l \in \{1,2,.
   	..,L\},
   \end{align}
   where $\rho_b^l\sim \mathcal{CN}(0,10^{-0.1\text{PL}_2 (d_{2,l})})$ denotes the complex gain of the channel that $\text{PL}_2 (d_{2,l})$ can be obtained as 
   \begin{align}
   \label{eq5_channel}
   \text{PL}_2 (d_{2,l}) \: [\text{dB}] = \mu_2 + 10 \kappa_2 \: \log_{10} (d_{2,l}) + \xi_2,
   \end{align}
   where $d_{2,l}$, $\mu_2$, and $\kappa_2$ denote the distance between the BS and the $l$-th IRS, constant path-loss term, and path-loss exponent, respectively. In addition, $ \xi_2 \sim \mathcal{N}(0,\sigma_{\xi_2}^2) $ where $\sigma_{\xi_2}^2$ is shadowing variance. 
   Furthermore, in \eqref{eq4_channel}, $\phi_l^b\in [0,\pi]$, and  $\psi_l^b\in [0,\pi]$ represent the elevation and azimuth angles of arrival for the $l$-th IRS and $\phi_b^l\in [0,2\pi]$ is the azimuth angle of departure for the BS. In \eqref{eq4_channel}, $\ab_b(\phi_b^l)$ is the normalized BS array response vector. $\ab_l(\phi_l^b,\psi_l^b)$ is also the normalized $l$-th IRS array response vector, that is denoted by 
   \begin{align}
   \label{eq6_channel}
   	\ab_l(\phi_l^b,\psi_l^b) = \ab_l^{az}(\phi_l^b,\psi_l^b) \otimes \ab_l^{el}(\phi_l^b),
   \end{align}
   where $\ab_l^{az}(\phi_l^b,\psi_l^b) $ and $\ab_l^{el}(\phi_l^b)$ are the horizontal and vertical array response vector of the $l$-th IRS, respectively. $\phi_l^b\in [0,\pi]$, and $\psi_l^b\in [0,\pi]$ also represent the elevation and azimuth angles of arrival for the $l$-th IRS, respectively. Furthermore, $\otimes$ denotes the Kronecker product operator.
   
   Similar to \eqref{eq1_channel} and \eqref{eq3_channel}, the $l$-th IRS-to-user and $l$-th IRS-to-eavesdropper channels can be respectively given by 
   \begin{align}
   \label{eq7_channel}
   	\hb_{l,k}^{\text{IRS-u}} = \frac{\sqrt N_l}{B} \sum_{\beta=1}^{B} \rho_{\beta,l}^k \ab_l(\phi_{\beta,l}^k,\psi_{\beta,l}^k),
   \end{align}
   and 
   \begin{align}
   \label{eq8_channel}
   	\gb_{l}^{\text{IRS-e }} = \frac{\sqrt N_l}{B} \sum_{\beta=1}^{B} \rho_{\beta,l}^e \ab_l(\phi_{\beta,l}^e,\psi_{\beta,l}^e),
   \end{align}
   where $\rho_{\beta,l}^k \sim \mathcal{CN}(0,10^{-0.1\text{PL}_1 (d_{3,l,k})})$ and $\rho_{\beta,l}^e \sim \mathcal{CN}(0,10^{-0.1\text{PL}_1 (d_{3,l,e})})$ are the complex gain of the $k$-th path of $l$-th IRS-to-user and $l$-th IRS-to-eavesdropper channel, respectively.
   $\phi_{\beta,l}^k\in [0,\pi]$ and $\psi_{\beta,l}^k\in [0,\pi]$ represent the elevation and azimuth angles of departure for the $l$-th IRS, respectively. In addition, $\phi_{\beta,l}^e\in [0,\pi]$ and $\psi_{\beta,l}^e\in [0,\pi]$ represent the elevation and azimuth angle of departure for the $l$-th IRS.
   $\ab_l(\phi_{\beta,l}^k,\psi_{\beta,l}^k)$ and $\ab_l(\phi_{\beta,l}^e,\psi_{\beta,l}^e)$ are the normalized array response vectors of the $l$-th IRS-user and the $l$-th IRS-eavesdropper paths, respectively. Definition of $\ab_l(\phi_{\beta,l}^k,\psi_{\beta,l}^k)$ and $\ab_l(\phi_{\beta,l}^e,\psi_{\beta,l}^e)$ are similar to \eqref{eq6_channel} based on Kronecker product of the $l$-th IRS horizontal and vertical array response vectors.
   

   \subsection{Signal Model}  
   The transmitted signal from the BS is  
  \begin{align}
 \xb =  \sum_{i=1}^K \wb_i s_i,
 \label{eq1}
  \end{align}
 where $s_i\sim \mathcal{CN}(0,1)$ and $\wb_i  \in \mathbb{C}^{M \times 1}$ are the intended signal for the $i$-th user and its associated BS beamforming vector, respectively.

	The received signal at the $k$-th legitimate user can be written as
\begin{align}
&y_{k} = \left(\sum_{l=1} ^ L {(\hb_{l,k}^{\text{IRS-u}})}^{H} \Thetab_l \Fb_l + {(\hb_{k}^{\text{BS-u}})}^H\right) \xb + n_{k} = \nonumber \\ &\left(\sum_{l} ^ L {(\hb_{l,k}^{\text{IRS-u}})} ^H \Thetab_l \Fb_l + {(\hb_{k}^{\text{BS-u}})}^H\right) \left(\wb_k s_k + \sum_{i \neq k}  \wb_i s_i \right) + n_{k}
\label{eq2}
\end{align}
where $\Thetab_l= \mathrm{diag} \{e^{j\theta_{l,1}}\:, ...\: ,\: e^{j\theta_{l,N}}\}  \in \mathbb{C}^{N_l \times N_l}$ is the passive beamforming matrix of the $l$-th IRS with $\theta_{l,n}$ as its phase shift value. To simplify the notation, in the rest of the paper, we will replace $e^{j\theta_{l,n}}$ with $\alpha_{l,n}$. Therefore, by defining $\boldsymbol{\alpha}_l = [\alpha_{l,1} ... ,\alpha_{l,N}]^T  \in \mathbb{C}^{N_l \times 1}$, we obtain
 $\Thetab_l= \mathrm{diag}  (\boldsymbol{\alpha}_l)$. $n_{k}\sim \mathcal{CN}(0,\sigma_k^2)$ is the additive white Gaussian noise at the $k$-th user receiver. 

Similar to \eqref{eq2}, the received signal at the eavesdropper is obtained as
\begin{align}
&y_{e} = \left(\sum_{l=1} ^ L {(\gb_{l}^{\text{IRS-e}})}^H \Thetab_l \Fb_l + {(\gb^{\text{BS-e}})}^H\right) \xb + n_{e}=\nonumber\\& \left(\sum_{l=1} ^ L {(\gb_{l}^{\text{IRS-e}})} ^H \Thetab_l \Fb_l + {(\gb^{\text{BS-e}})}^H\right) \left(\sum_{i=1}^K \wb_i s_i\right) + n_{e}
\label{eq3}
\end{align}
where $n_{e}\sim \mathcal{CN}(0,\sigma^2_e)$ is the additive white Gaussian noise at the eavesdropper receiver.



According to \eqref{eq2} and \eqref{eq3}, signal-to-noise-plus-interference ratio (SINR) for the $k$-the legitimate user ($\Gamma_{k}$) and eavesdropper ($\Gamma_{e}$) are obtained as follows
  
\begin{align}
&\Gamma_{k}= \frac{\bigg|\left( \displaystyle\sum_{l=1} ^ L  {(\hb_{l,k}^{\text{IRS-u}})} ^H \Thetab_l \Fb_l + {(\hb_{k}^{\text{BS-u}})}^H\right) \wb_k\bigg|^2}  {\displaystyle\sum_{i=1, i \neq k}^K\bigg| \left(\displaystyle\sum_{l=1} ^ L {(\hb_{l,k}^{\text{IRS-u}})} ^H \Thetab_l \Fb_l + {(\hb_{k}^{\text{BS-u}})}^H\right) \wb_i\bigg|^2 + \sigma_k^2},
\label{eq4}
\end{align}
\begin{align}
&\Gamma_{e}= \frac{\bigg|\left(\displaystyle\sum_{l=1} ^ L  {(\gb_{l}^{\text{IRS-e}})}^H \Thetab_l \Fb_l + {(\gb^{\text{BS-e}})}^H\right) \wb_k\bigg|^2}  {\displaystyle\sum_{i=1, i \neq k} ^K \bigg| \left(\displaystyle\sum_{l=1} ^ L  {(\gb_{l}^{\text{IRS-e}})}^H \Thetab_l \Fb_l + {(\gb^{\text{BS-e}})}^H\right) \wb_i\bigg|^2 + \sigma_e^2}.
\label{eq5}
\end{align} 
 To obtain these equations, we leverage the fact that all channels between two different network elements are independent. 
 
Using \eqref{eq4} and \eqref{eq5}, the secrecy rate for the $k$-th user can be obtained as 
\begin{align}
&R_{s,k} = \bigg[\log_2 (1+ \Gamma_{k}) - \log_2 (1+\Gamma_{e})\bigg]^+ = 
\nonumber \\ &\Bigg[\log_2 \bigg(1+\frac{\bigg|\big(\displaystyle\sum_{l=1} ^ L  {(\hb_{l,k}^{\text{IRS-u}})}^H \Thetab_l \Fb_l + {(\hb_{k}^{\text{BS-u}})}^H\big) \wb_k\bigg|^2}  {\displaystyle\sum_{i=1, i \neq k}^K \bigg|\big(\displaystyle\sum_{l=1} ^ L  {(\hb_{l,k}^{\text{IRS-u}})} ^H \Thetab_l \Fb_l + {(\hb_{k}^{\text{BS-u}})}^H\big) \wb_i\bigg|^2 + \sigma_{k}^2}\bigg)- \nonumber \\&
\log_2 \bigg(1+\frac{\bigg|\big(\displaystyle\sum_{l=1} ^ L {(\gb_{l}^{\text{IRS-e}})} ^H \Thetab_l \Fb_l + {(\gb^{\text{BS-e}})}^H\big) \wb_k\bigg|\bigg|^2}  {\displaystyle\sum_{i=1, i \neq k} ^K \bigg| \big(\displaystyle\sum_{l=1} ^ L  {(\gb_{l}^{\text{IRS-e}})} ^H \Thetab_l \Fb_l + {(\gb^{\text{BS-e}})}^H\big) \wb_i\bigg|^2 + \sigma_e^2}\bigg)\Bigg]^+,
\label{eq6}
\end{align}
where $[z]^+=\text{max}(z,0)$, meaning that only the positive values of secrecy rate are acceptable.


\section{Problem Formulation}
\label{sec3}

	As secrecy rate mainly focuses on improving legitimate users' rate while deteriorating eavesdropper's rate, we choose it as our objective function of the optimization problem. In practical multi-user systems, the eavesdropper can only be close to one of the users. In such a case, the other users have higher secrecy rate than the user close to the eavesdropper. Therefore, utilizing the sum of users' secrecy rate as the objective function of the optimization problem leads to the secrecy rate of the user close to the eavesdropper be ignored.
	In other words, in order to increase the sum of the secrecy rates, the network resources are allocated to users who are farther away from the eavesdropper, making the security of the network questionable. To cope with this issue, weighted secrecy rate optimization or max-min fairness criterion can be employed. Determining weighting factor in the former approach can be a difficult and challenging task in practice. Therefore, in this paper, we choose the max-min fairness in our problem formulation by jointly optimizing active and passive beamforming, which will be discussed in the following.
	
	The secrecy rate optimization problem can be expressed as     
	\begin{subequations}
	\begin{align}
		 \qquad \mathcal{P}:\: \: &\max_{\wb_1,...,\wb_K,\boldsymbol{\alpha}_1,...,\boldsymbol{\alpha}_L} \: \min_{k} \: \:  R_{s,k}, \\ \label{eq7b}&
	\qquad\text{s.t.} \: \: \:  \sum_{i=1}^K||\wb_i||_2^2\le P_{\text{BS}}, \\ \label{eq7c}& \qquad
	|\alpha_{l,n}|= 1 , \: \:   \:  l \in \{1,...,L\}, \: n \in \{1,...,N_l\}. \\ &\qquad \label{eq7d}
	\angle{\alpha_{l,n}} \in \mathcal{Z}	= \{0, \Delta \theta_l\, ... , (N_l-1) \Delta \theta_l \}, \: \:   \: \nonumber \\ &\qquad \qquad \qquad \quad \quad l \in \{1,...,L\}, \: n \in \{1,...,N_l\}.
	    	\end{align}  
	    \end{subequations}
   The objective function of the optimization problem is to maximize the minimum secrecy rate among users. The first constraint \eqref{eq7b} limits the maximum transmission power of the BS to be $ P_{\text{BS}}$. The second constraint \eqref{eq7c} states the unitary gain of IRS phase shifts which reflects that the IRS elements are passive without power amplification ability. The final constraint overcomes the hands-on difficulties regarding the continuous-phase shift implementation in IRSs. In other words, IRSs phase shifts are considered being able to get some discrete values according to set $\mathcal{Z}$. Here, $\Delta \theta_l$ is the phase resolution of the $l$-th IRS which equals to $ \Delta \theta_l = \frac {2\pi} {N_l}$.\\ 
 Based on the problem formulation, $\mathcal{P}$ is not only a mixed integer-continuous optimization problem, but also includes highly coupled optimization variables, due to the discrete values of $\boldsymbol{\alpha}$, continuous values of beamfoming vectors, and high coupling among them in objective function. Reaching an optimal solution of such a strong non-convex problem is ridiculously complicated. Although, the values of $\boldsymbol{\alpha}$ are in a finite set and can be found by an exponential complexity search, the coupling relation among $\boldsymbol{\alpha}$, and $\wb$ makes this search inapplicable. Thus, apart from the complexity, there is no-straightforward method to find an optimal solution. One common method to make this kind of problems tractable is to break the main problem into some sub-problems according to BCD approach and then solve the main problem in an iterative-manner. SDP in combination with rank-one relaxation and Gaussian randomization is one of the major proposed methods for solving each sub-problems in the literature \cite{jadid6, AH10, jadid7, jadid8, jadid1, jadid2}. One of the main problems of such approaches is the convergence and also achieving an approximated solution. In particular, because of using Gaussian randomization and/or EVD method to refine a rank-one method, it is required a large number of iterations for convergence, especially in our multi-IRS proposed system with a large number of phase shifts and BS antennas. Furthermore, there is no-strong guarantee without any tricks for BCD convergence in such a case that all the results are approximated by the Gaussian randomization and/or EVD method. Considering these facts, in the following, we propose our algorithm to cope with these problems with an affordable computational complexity.
\section{Proposed Algorithm}
\label{sec4}
  Taking a closer look at the non-convex problem $\mathcal{P}$, it is found that there are two disjoint subsets of optimization variables, namely $\{\wb_1,...,\wb_K\}$ and $\{\boldsymbol{\alpha}_1,...,\boldsymbol{\alpha}_L\}$ with their own independent constraints. According to the aforementioned discussion, one way to deal with problem $\mathcal{P}$ is the BCD approach. However, dealing with discrete variables as well as absolute value constraint \eqref{eq7c} is very challenging indeed.
   On way to cope the former problem is continuous relaxation with an affordable computational complexity. As in IRSs, a large number of reflecting elements are widely used, continuous relaxation can lead to a low quantization error. Thus, in this paper, we use continuous relaxation to deal with discrete variables which replaces \eqref{eq7d} by the following convex constraint
  
  \begin{align}
  \label{relax1}
  & 0\leq\angle{\alpha_{l,n}} \leq2\pi,  \: l \in \{1,...,L\}, \: n \in \{1,...,N_l\}.
  \end{align} 
  
  Actually, we relax the $\boldsymbol{\alpha}$-relevant discrete phase constraint into convex alternatives and change the problem $\mathcal{P}$ from the mixed integer-continuous into a continuous one.\\ To resolve the mentioned issue regarding \eqref{eq7c}, convex-concave decomposition approach \cite{convex_concave_decom} leads to
  	\begin{subequations}
  	\begin{align}
  	&|\alpha_{l,n}|^2 \leq 1 , \: \:   \:  \forall l,n, \\ \label{alpha_concave1}
  	&|\alpha_{l,n}|^2 \geq 1-\epsilon_{l,n} , \: \:   \:  \forall l,n, \\\label{relax_penalty1}
  	& \epsilon_{l,n} = 0, \: \:  \:  \forall l,n, 
  	\end{align}
  \end{subequations}
  where $\epsilon_{l,n}$ is a slack variable models the equality constraint violation. In this situation, as the formed constraints are too strict, applying the concept of penalty function is highly suggested to broaden the feasible set \cite{penalty1, penalty2}. Employing this approach, approximated formulation of $\mathcal{P}$ with a larger feasible set and also the constraint violation penalty term is 
\begin{subequations}
	\begin{align}
	\quad \mathcal{P}':\: \: &\max_{\wb_1,...,\wb_K,\boldsymbol{\alpha}_1,...,\boldsymbol{\alpha}_L, \epsilon_{l,n}} \: \min_{k} \: \:  R_{s,k} + P_e\sum_{l,n}\epsilon_{l,n}, \\ &
	\quad\text{s.t.} \: \: \:  \eqref{eq7b},\eqref{relax1} \\ 
	&\quad|\alpha_{l,n}|^2 \leq 1 , \: \:   \:  \forall l,n, \\ \label{alpha_concave2}
	&\quad|\alpha_{l,n}|^2 \geq 1-\epsilon_{l,n} , \: \:   \:  \forall l,n, \\\label{relax_penalty2}
	&\quad 0\leq\epsilon_{l,n}\leq 1, \: \:  \:  \forall l,n, 
	\end{align}  
\end{subequations}
Here, according to the concept of penalty method, \eqref{relax_penalty1} is relaxed to \eqref{relax_penalty2} and penalty function $\sum_{l,n}\epsilon_{l,n}$ with penalty coefficient $P_e<0$ is added to objective function in order to penalize the constraint violations. Due to $0\leq\epsilon_{l,n}\leq 1, \forall l,n $, the violations less than 1 are absolutely important, thus $l_1$-norm is considered as the penalty function. Unfortunately, penalty function method as well as continuous relaxation approximate the original optimization problem $\mathcal{P}$ by increasing the range of feasible set in $\mathcal{P}'$. As a result, the values of optimization variables are less likely to be feasible in  $\mathcal{P}$. The feasible solutions are derived based on the proposed mapping algorithm in sub-section \ref{mapping}. However, as simulation results will show, this concept still leads to high-quality results compared with proposed benchmarks.\\
  
 In the next step, using BCD approach, the optimization problem $\mathcal{P}'$ is at first decomposed into two disjoint sub-problems, which are then solved in an iterative-manner until convergence. In the following, these two sub-problems, the proposed approaches and the overall BCD-based algorithm to solve the relaxed problem $\mathcal{P}'$ will be presented in detail. Furthermore, the mapping algorithm to navigate to feasible points of $\boldsymbol{\alpha}$ will also be explicated.  
	\subsection{Active Beamforming Optimization}
Assuming that $\Thetab$ and $\epsilon_{l,n}, \forall l,n$ in $\mathcal{P}'$ are fixed, the sub-problem 1 is obtained as
		\begin{subequations}
		\label{active_main}
		\begin{align}
		&	\mathcal{P}1: \max_{\wb_1,...,\wb_K} \: \min_{k} \: \: R_{s,k} + \text{Ct}\\& 
		\text{s.t.} \: \: \: \sum_{i=1}^K||\wb_i||_2^2\le P_{\text{BS}}.
		\end{align}
	\end{subequations}
	where
	\begin{align}
	& \qquad \qquad \qquad \qquad \qquad R_{s,k} =\nonumber\\& \log_2\Bigg[\frac{\bigg(\displaystyle\sum_{i=1}^K\bigg|(\displaystyle\sum_{l=1} ^ L  {(\hb_{l,k}^{\text{IRS-u}})} ^H \Thetab_l \Fb_l + {(\hb_{k}^{\text{BS-u}})}^H) \wb_i\bigg|^2 + \sigma_k^2\bigg)}  {\bigg(\displaystyle\sum_{i \neq k} \bigg|(\displaystyle\sum_{l=1} ^ L  {(\hb_{l,k}^{\text{IRS-u}})} ^H \Thetab_l \Fb_l + {(\hb_{k}^{\text{BS-u}})}^H) \wb_i\bigg|^2 + \sigma_k^2\bigg)}\Bigg] -\nonumber \\&\quad
	\log_2 \Bigg[\frac  {\bigg(\displaystyle\sum_{i = 1}^K \bigg|(\displaystyle\sum_{l=1} ^ L  {(\gb_{l}^{\text{IRS-e}})} ^H \Thetab_l \Fb_l + {(\gb^{\text{BS-e}})}^H) \wb_i\bigg|^2 + \sigma_e^2\bigg)} {\bigg(\displaystyle\sum_{i \neq k}\bigg|(\displaystyle\sum_{l=1} ^ L  {(\gb_{l}^{\text{IRS-e}})} ^H \Thetab_l \Fb_l + {(\gb^{\text{BS-e}})}^H) \wb_i\bigg|^2+\sigma_e^2\bigg)} \Bigg].
	\end{align}
 Here in \eqref{active_main}, $\text{Ct}=P_e\sum_{l,n}\epsilon_{l,n}$ is a constant value in this sub-problem which does not influence the optimum values of active beamforming vectors. Thus, for the reset of the active beamforming optimization, it will be ignored. Since the BS is capable to amplify the transmitted signal or in other words optimize both amplitude and phase of the beamforming coefficients, this problem is named as active beamforming. The objective function in $\mathcal{P}1$ is still non-convex, thus we reformulate it according to the following proposition. 

\noindent{\bf Proposition 1.} \emph{The optimization problem $\mathcal{P}1$ can be reformulated as the following optimization problem.}
\begin{subequations}
	\label{eq_reform1}
	\begin{align}
		& \max_{\wb_1,...,\wb_K,p_k,q_k,r,s_k,a}  \:\min_{k} \: \: -r+a \\ \label{eqt1}
		& \text{s.t.} \: \: \: \sum_{i=1}^K||\wb_i||_2^2\le P_{\text{BS}},\\ \label{eqt2}
		& \quad\quad p_k-q_k+s_k \ge a , \: \: \forall k \\ \label{eqt3p}
		&\quad\quad \sum_{i=1}^K |\boldsymbol{\xi}_k^T \wb_i | ^2+\sigma_r^2\ge 2^{p_k}, \: \: \forall k \\ \label{eqt3}
		&\quad\quad \sum_{i \neq k} | \boldsymbol{\xi}_k^T \wb_i | ^2+\sigma_r^2\le 2^{q_k} , \: \: \forall k \\ \label{eqt4}
		&\quad\quad \sum_{i=1}^K |\boldsymbol{\rho}^T \wb_i|^2+\sigma_e^2\le 2^{r}, \: \: \forall k \\ \label{eqt5}
		&\quad\quad \sum_{i \neq k} |\boldsymbol{\rho}^T \wb_i| ^2+\sigma_e^2\ge 2^{s_k}. \: \: \forall k 	
	\end{align}
\end{subequations} 
where $ \boldsymbol{\xi}_k^T =  \big(\sum_{l=1} ^ L  {(\hb_{l,k}^{\text{IRS-u}})} ^H \Thetab_l \Fb_l + {(\hb_{k}^{\text{BS-u}})}^H\big) \in \mathbb{C}^{1 \times M}$ , and $\boldsymbol{\rho}^T =  \big(\sum_{l=1} ^ L {(\gb_{l}^{\text{IRS-e}})} ^H \Thetab_l \Fb_l + {(\gb^{\text{BS-e}})}^H\big) \in \mathbb{C}^{1 \times M}$. 
\begin{IEEEproof} 
Defining $ \boldsymbol{\xi}_k^T =  \big(\sum_{l=1} ^ L  {(\hb_{l,k}^{\text{IRS-u}})} ^H \Thetab_l \Fb_l + {(\hb_{k}^{\text{BS-u}})}^H\big)$, and $\boldsymbol{\rho}^T =  \big(\sum_{l=1} ^ L {(\gb_{l}^{\text{IRS-e}})} ^H \Thetab_l \Fb_l + {(\gb^{\text{BS-e}})}^H\big)$, the following simplifications can be carried out,
\begin{align}
 \bigg|\big(\sum_{l=1} ^ L  {(\hb_{l,k}^{\text{IRS-u}})} ^H \Thetab_l \Fb_l + {(\hb_{k}^{\text{BS-u}})}^H\big) \wb_i\bigg|^2 =  | \boldsymbol{\xi}_k^T \wb_i |^2, \label{eq10}\\
 \bigg|\big( \sum_{l=1} ^ L  {(\gb_{l}^{\text{IRS-e}})} ^H \Thetab_l \Fb_l + {(\gb^{\text{BS-e}})}^H\big) \wb_i\bigg|^2 = 
 |\boldsymbol{\rho}^T \wb_i|^2.
 \label{eq11}
\end{align}
\\
Considering \eqref{eq10} and \eqref{eq11}, $\mathcal{P}1$ is reformulated as
 \begin{subequations}
	\begin{align}
	&	\mathcal{P}1: \max_{\wb_1,...,\wb_K} \: \min_{k} \: \: 
	\log_2\Bigg(\frac{\displaystyle\sum_{i=1}^K| \boldsymbol{\xi}_k^T \wb_i | ^2 + \sigma_k^2}  {\displaystyle\sum_{i \neq k} |\boldsymbol{\xi}_k^T \wb_i|^2 + \sigma_k^2}\Bigg) -\nonumber \\&
	\qquad\qquad\qquad\qquad\qquad\qquad\log_2 \Bigg(\frac  {\displaystyle\sum_{ i = 1}^K |\boldsymbol{\rho}^T \wb_i |^2 + \sigma_e^2} {\displaystyle\sum_{i \neq k}|\boldsymbol{\rho}^T \wb_i | ^2+\sigma_e^2} \Bigg),\\ 
	&	\qquad\qquad\qquad\qquad\text{s.t.} \: \: \: \sum_{i=1}^K||\wb_i||_2^2\le P_{\text{BS}}.
	\end{align}
	\label{eq12}
\end{subequations}
After that, by introducing auxiliary variables $ 2^{p_k}=\sum_{i=1}^K|\boldsymbol{\xi}_k^T \wb_i| ^2 + \sigma_k^2$,  $2^{q_k}=\sum_{i \neq k} |\boldsymbol{\xi}_k^T \wb_i|^2 + \sigma_k^2$,\\ $2^{r}=\sum_{ i = 1}^K |\boldsymbol{\rho}^T \wb_i|^2 + \sigma_e^2$ and $2^{s_k}=\sum_{i \neq k}|\boldsymbol{\rho}^T \wb_i|^2+\sigma_e^2 $ and substituting into \eqref{eq12}, we arrive at the following equivalent optimization problem.
 \begin{subequations}
 	\label{eq12_v1}
 	\begin{align}
 		& \max_{\wb_1,...,\wb_K,p_k,q_k,r,s_k}  \:\min_{k} \: \: p_k -q_k -r+s_k \\ \label{eq9b}
 		& \text{s.t.} \: \: \: \sum_{i=1}^K||\wb_i||_2^2\le P_{\text{BS}},\\ \label{eq9c}
 		&\quad\quad \sum_{i=1}^K |\boldsymbol{\xi}_k^T \wb_i|^2+\sigma_k^2\ge 2^{p_k}, \: \: \forall k \\ \label{eq9f}
 		&\quad\quad \sum_{i \neq k} |\boldsymbol{\xi}_k^T \wb_i| ^2+\sigma_k^2\le 2^{q_k} , \: \: \forall k \\ \label{eq9g}
 		&\quad\quad \sum_{i=1}^K |\boldsymbol{\rho}^T \wb_i|^2+\sigma_e^2\le 2^{r}, \: \: \forall k \\ \label{eq9h}
 		&\quad\quad \sum_{i \neq k} |\boldsymbol{\rho}^T\wb_i| ^2+\sigma_e^2\ge 2^{s_k}. \: \: \forall k  
 	\end{align}
 \end{subequations}
Again, by introducing another auxiliary variable $a= \underset{k}{\text{min}}(p_k - q_k + s_k)$, \eqref{eqt2} constitutes, the optimization problem \eqref{eq12_v1} reformulated as \eqref{eq_reform1} and the proof is completed.   
\end{IEEEproof}                                 
The objective function in \eqref{eq_reform1} is an affine one, thus it is concave. The constraints \eqref{eqt1} and \eqref{eqt2} also form convex sets. As it has been expected, there are some non-convex constraints involving \eqref{eqt3p}, \eqref{eqt3}, \eqref{eqt4} and \eqref{eqt5}. With respect to these constraints, it is found that both sides of the inequalities of all constraints are in convex forms, therefore SCA can be employed to tackle their non-convexity. But before applying SCA approach, it is required to introduce below auxiliary lemma which will be helpful in the process of applying SCA.\\
\noindent{\bf Lemma 1}.
For a general quadratic function in form of $y=|\ab^T\xb+b|^2$ where $\ab=\ab^{\Re}+j\ab^{\Im}  \in \mathbb{C}^{N \times 1}$, $\xb=\xb^{\Re}+j\xb^{\Im} \in \mathbb{C}^{N \times 1}$, the partial Gradients w.r.t. the real and imaginary parts of $\xb$ (i.e. $\xb^{\Re}$ and $\xb^{\Im}$), are acquired according to the following equation
 \begin{align}
 \label{lemma1_main_eq}
  \begin{bmatrix} 
   \dot{\yb}^{\Re}\\
   \dot{\yb}^{\Im}\
 \end{bmatrix}
 =
 \begin{bmatrix} 
 \ab^{\Re} & \ab^{\Im} \\
 -\ab^{\Im} & \ab^{\Re} 
 \end{bmatrix}
 \begin{bmatrix} 
\Re\big\{\ab^T\xb+b\}  \\
\Im\big\{\ab^T\xb+b\}  
\end{bmatrix} 			
 \end{align} 
 where $ \dot{\yb}^{\Re}$ and $\dot{\yb}^{\Im}$ are partial Gradients of $\mathbf{y}$ w.r.t. $\xb^{\Re}$ and $\xb^{\Im}$, respectively. Here, block matrix notation is also used to demonstrate the partial Gradients. Furthermore, $\Re\big\{.\big\}$ and  $\Im\big\{.\big\}$ are operators returning the real and imaginary part of complex vectors and matrices.
 \begin{IEEEproof}
See Appendix \ref{Apen1}.	
 \end{IEEEproof}  
 Considering Lemma 1, subsequent Lemma articulates employing SCA to create approximated convex constraints. 
 
\noindent{\bf Lemma 2}.
 \eqref{eqt3p},\eqref{eqt3}, \eqref{eqt4}, and \eqref{eqt5}, can be approximated as the following convex constraints by using the SCA approach.
\begin{subequations}
	\begin{align}
	\label{s1sca1}
	& \sum_{i=1}^K |\boldsymbol{\xi}_k^T \bar{\wb}_i|^2+ 2 \sum_{i=1}^K \big[({\boldsymbol{\gamma}}_{i,k}^\Re)^T(\wb_i^{\Re} - \bar{\wb}_i^{\Re})+\nonumber\\&\qquad\qquad\qquad\qquad \quad({\boldsymbol{\gamma}}_{i,k}^\Im)^T(\wb_i^{\Im} - \bar{\wb}_i^{\Im})\big] +\sigma_k^2\ge 2^{p_k}, \: \: \forall k \\ \label{s1sca2}
	& \sum_{i \neq k} |\boldsymbol{\xi}_k^T\wb_i|^2+\sigma_k^2\le 2^{\bar{q}_k} + \ln(2)2^{\bar{q}_k} (q_k - \bar{q}_k) , \: \: \forall k \\ 
	\label{s1sca3}
	&\sum_{i=1}^K |\boldsymbol{\rho}^T \wb_i|^2+\sigma_e^2\le 2^{\bar{r}} + \ln(2)2^{\bar{r}} (r - \bar{r}), \: \: \forall k \\ 
	\label{s1sca4}
	&\sum_{i \neq k} |\boldsymbol{\rho}^T \bar{\wb}_i|^2+2\sum_{i \neq k} \big[({\boldsymbol{\kappa}}_{i}^\Re)^T(\wb_i^{\Re} - \bar{\wb}_i^{\Re})+\nonumber\\&\qquad\qquad\qquad\qquad \quad({\boldsymbol{\kappa}}_{i}^\Im)^T(\wb_i^{\Im} - \bar{\wb}_i^{\Im})\big]+\sigma_e^2\ge 2^{s_k}. \: \: \forall k,   
	\end{align}
\end{subequations}
where $\wb_i = \wb_i^{\Re} + j\wb_i^{\Im}$, $\boldsymbol{\xi}_k = \boldsymbol{\xi}_k^{\Re} + j\boldsymbol{\xi}_k^{\Im}$, $\boldsymbol{\rho} = \boldsymbol{\rho}^{\Re} + j\boldsymbol{\rho}^{\Im}$. $\bar{\wb}_i^{\Re}$, $\bar{\wb}_i^{\Im}$, $\bar{q}_k$, and $\bar{r}$ are also the points where the first-order Taylor approximations are made. Furthermore, ${\boldsymbol{\gamma}}_{i,k}^\Re, {\boldsymbol{\gamma}}_{i,k}^\Im, {\boldsymbol{\kappa}}_{i}^\Re,$ and ${\boldsymbol{\kappa}}_{i}^\Im \in \mathbb{C}^{M\times1}$ are denoted as
\begin{subequations}
\begin{align}
&\begin{bmatrix} 
	{\boldsymbol{\gamma}}_{i,k}^\Re\\
	{\boldsymbol{\gamma}}_{i,k}^\Im
\end{bmatrix}
=
\begin{bmatrix} 
	\boldsymbol{\xi}_k^{\Re} & \boldsymbol{\xi}_k^{\Im} \\
	-\boldsymbol{\xi}_k^{\Im} & \boldsymbol{\xi}_k^{\Re} 
\end{bmatrix}
\begin{bmatrix} 
	\Re\big\{\boldsymbol{\xi}_k^T\bar{\wb}_i\big\}  \\
	\Im\big\{\boldsymbol{\xi}_k^T\bar{\wb}_i\big\}  
\end{bmatrix} \\
&\begin{bmatrix} 
{\boldsymbol{\kappa}}_{i}^\Re\\
{\boldsymbol{\kappa}}_{i}^\Im
\end{bmatrix}
=
\begin{bmatrix} 
\boldsymbol{\rho}^{\Re} & \boldsymbol{\rho}^{\Im} \\
-\boldsymbol{\rho}^{\Im} & \boldsymbol{\rho}^{\Re} 
\end{bmatrix}
\begin{bmatrix} 
\Re\big\{\boldsymbol{\rho}^T\bar{\wb}_i\big\}  \\
\Im\big\{\boldsymbol{\rho}^T\bar{\wb}_i\big\}  
\end{bmatrix}	
\end{align}
\end{subequations}
\begin{IEEEproof} 
As for convex functions, the existence of first-order Taylor expansions creates themselves lower bounds, Taylor expansions of $2^{q_k}$ and $2^{r}$ around $\bar{q_k}$ and $\bar{r}$ replace \eqref{eqt3} and \eqref{eqt4} with convex constraints \eqref{s1sca2} and \eqref{s1sca3}, respectively. Contrary to $2^{q_k}$ and $2^{r}$, $\sum_{i=1}^K | \boldsymbol{\xi}_k^T \wb_i |^2$ and $\sum_{i \neq k}|\boldsymbol{\rho}^T\wb_i |^2$ are not differentiable w.r.t. the complex vector $\wb_i$. Thus, by defining $\wb_i = \wb_i^{\Re} + j\wb_i^{\Im}$ and using the first order Taylor expansions of equivalent inequalities based on real and imaginary parts of $\wb_i$, according to Lemma 1, \eqref{eqt3p} and \eqref{eqt5} can be replaced by \eqref{s1sca1} and \eqref{s1sca4}, respectively.
\end{IEEEproof}	 
According to Lemma 2 in conjunction with the SCA approach, $\mathcal{P}1$ can be approximated with a smaller feasible set than the original problem as the following iterative convex problem
  \begin{subequations}
 	\begin{align}
 	& \mathcal{P}1.1 \max_{\wb_1,...,\wb_K,p_k,q_k,r,s_k,a}  \: \: -r+a \\ \label{pow_const}
 & \text{s.t.} \: \: \: \sum_{i=1}^K||\wb_i||_2^2\le P_{\text{BS}},\\ 
 & p_k-q_k+s_k \ge a , \: \: \forall k \\
 & \sum_{i=1}^K|\boldsymbol{\xi}_k^T \bar{\wb}_i^{(t_1)}| ^2+ 2 \sum_{i=1}^K \big[({\boldsymbol{\gamma}}_{i,k}^{\Re(t_1)})^T(\wb_i^{\Re} - \bar{\wb}_i^{\Re{(t_1)}})+\nonumber\\&\qquad\qquad\qquad({\boldsymbol{\gamma}}_{i,k}^{\Im(t_1)})^T(\wb_i^{\Im} - \bar{\wb}_i^{\Im(t_1)})\big] +\sigma_k^2\ge 2^{p_k}, \: \: \forall k \\ 
 & \sum_{i \neq k} |\boldsymbol{\xi}_k^T \wb_i |^2+\sigma_k^2\le 2^{\bar{q}_k^{(t_1)}} + \ln(2)2^{\bar{q}_k^{(t_1)}} (q_k - \bar{q}_k^{(t_1)}) , \: \: \forall k \\ 
 &\sum_{i=1}^K |\boldsymbol{\rho}^T \wb_i |^2+\sigma_e^2\le 2^{\bar{r}^{(t_1)}} + \ln(2)2^{\bar{r}^{(t_1)}} (r - \bar{r}^{(t_1)}), \: \: \forall k \\ 
 &\sum_{i \neq k} |\boldsymbol{\rho}^T \bar{\wb}_i^{(t_1)}| ^2+2\sum_{i \neq k} \big[({\boldsymbol{\kappa}}_{i}^{\Re{(t_1)}})^T(\wb_i^{\Re} - \bar{\wb}_i^{\Re{(t_1)}})+\nonumber\\&\qquad\qquad\qquad({\boldsymbol{\kappa}}_{i}^{\Im^{(t_1)}})^T(\wb_i^{\Im} - \bar{\wb}_i^{\Im{(t_1)}})\big]+\sigma_e^2\ge 2^{s_k}, \: \: \forall k,
 	\end{align}
 \end{subequations}
 where the optimization problem is solved at the $t_1$-th iteration. The proposed algorithm for solving $\mathcal{P}1$ is summarized in \textbf{Algorithm 1}.

 	\begin{algorithm}[H]
 		\caption{Proposed algorithm for solving $\mathcal{P}1$.}
 		\begin{algorithmic}[1]
 			\renewcommand{\algorithmicensure}{\textbf{Outputs:}}
 			\renewcommand{\algorithmicrequire}{\textbf{Initialization:}} 				
 			\ENSURE  $ {\wb_1},...,{\wb_K}. $		
 			\REQUIRE  			  ${t_{1,\text{max}}},$${\varepsilon_1},$${\boldsymbol{\alpha}_1},...,{\boldsymbol{\alpha}_L}$, $\epsilon_{l,n}, \forall l,n$, $t_1= 0,$$ \bar{q}_k^{(0)},$ $\forall k , \bar{r}^{(0)}, \bar{\wb}_1^{(0)},...,\bar{\wb}_K^{(0)}, $ such that the sub-problem is feasible. 
 			\STATE \textbf{Repeat}
 			\STATE Solve the problem $\mathcal{P}1.1$ to obtain the optimal solutions for $\wb_1^{(t_1)},...,\wb_K^{(t_1)}$.
 			\STATE Set $ t_1=t_1+1 $.
 			\STATE Update $\bar{\wb}_1^{(t_1+1)}=\wb_1^{(t_1)}, ..., \bar{\wb}_K^{(t_1+1)}=\wb_K^{(t_1)}$, $\bar{q}_k^{(t_1+1)}=q_k^{(t_1)}, \forall k , \: \bar{r}^{(t_1+1)}=r^{(t_1)}$, 
 			\STATE \textbf{Until} $ \left| {\frac{{\underset{k}{\text{min}} {R'}_{s,k}^{(t_1)} -\underset{k}{\text{min}} {R'}_{s,k}^{(t_1-1)}}}{{\underset{k}{\text{min}}{R'}_{s,k}^{(t_1-1)}}}} \right| \le {\varepsilon_1 } $ or $t_1 = t_{1,\text{max}}$,
 		\end{algorithmic}
 	\end{algorithm}

In \textbf{Algorithm 1}, ${t_{1,max}}, {\varepsilon_1},{\boldsymbol{\alpha}_1},...,{\boldsymbol{\alpha}_L}$ are the maximum number of iterations for the algorithm, maximum relative error of objective function between two consecutive outputs, and $L$ IRS phase profiles, respectively. The objective function of this sub-problem at the $t_1$-th iteration is defined as  $R_{s,k}^{'(t_1)} = R_{s,k}^{(t_1)} +  \text{Ct}$. BS beamforming vectors are the outputs of this algorithm. In the initialization step, we randomly and independently generate the vectors ${\bar{\wb}_1}^{(0)},...,{\bar{\wb}_K}^{(0)}$ according to $\mathcal{CN}(0,1)$ and then normalize their values such that \eqref{pow_const} becomes an active constraint. Considering these initial values, $\bar{q_k}^{(0)}, \forall k$ and $\bar{r}^{(0)}$ are calculated according to the following equations
  \begin{subequations}
 \begin{align}
 & {\bar{q}_k^{(0)}}=\log_2\bigg(\sum_{i=1}^K|\boldsymbol{\xi}_k^T \bar{\wb}_i^{(0)}|^2 + \sigma_k^2\bigg),  \forall k\\
 &{\bar{r}^{(0)}}=\log_2\bigg(\sum_{i = 1}^K |\boldsymbol{\rho}^T \bar{\wb}_i^{(0)}|^2 + \sigma_e^2\bigg),
 \end{align}
 \end{subequations}
where these values are achieved according to the definition of auxiliary variables in Proposition 1.
 For the subsequent iterations (except the first one), the initial values of $\bar{\wb}_1,...,\bar{\wb}_K$, $\bar{q}_k, \forall k$, and $\bar{r}$ are attained from the previous iteration. For example at the $(t_1+1)$-th iteration, the initial values are obtained as $\bar{\wb}_1^{(t_1+1)}=\wb_1^{(t_1)},...,\bar{\wb}_K^{(t_1+1)}=\wb_K^{(t_1)}$, $\bar{q}_k^{(t_1+1)} = {q}_k^{(t_1)} $ and $\bar{r}^{(t_1+1)} = {r}^{(t_1)}$. It is obvious that these values also belong to the feasible set of the new problem at the $(t_1+1)$-th iteration.
 
\subsection{Passive Beamforming Optimization}
In the second sub-problem, given $\wb_i, \forall i$, the optimization problem $\mathcal{P}'$ can be formulated as
\begin{subequations}
	\label{P2}
	\begin{align}
	\label{eq.objective.P2}
	&	\mathcal{P}2: \max_{\boldsymbol{\alpha},\epsilon_{l,n}, \forall l,n} \: \min_{k} \: \: R_{s,k} + P_e\sum_{l = 1} ^{NL}\epsilon_{l,n},\\ 
	&\label{eq.constraint.P2} 
	\text{s.t.} \: \: \: \:\quad|\alpha_{l,n}|^2 \leq 1 , \: \:   \:  \forall l,n, \\ \label{eq.constraint.P2_1}
	&\qquad\quad|\alpha_{l,n}|^2 \geq 1-\epsilon_{l,n} , \: \:   \:  \forall l,n, \\\label{eq.constraint.P2_2}
	&\qquad\quad 0\leq\epsilon_{l,n}\leq 1, \: \:  \:  \forall l,n, 
	\end{align}
\end{subequations}
where
\begin{align}
& \qquad \qquad \qquad \qquad \qquad R_{s,k} =\nonumber\\& \log_2\Bigg[\frac{\bigg(\displaystyle\sum_{i=1}^K\bigg|(\displaystyle\sum_{l=1} ^ L  {(\hb_{l,k}^{\text{IRS-u}})} ^H \Thetab_l \Fb_l + {(\hb_{k}^{\text{BS-u}})}^H) \wb_i\bigg|^2 + \sigma_k^2\bigg)}  {\bigg(\displaystyle\sum_{i \neq k} \bigg|(\displaystyle\sum_{l=1} ^ L  {(\hb_{l,k}^{\text{IRS-u}})} ^H \Thetab_l \Fb_l + {(\hb_{k}^{\text{BS-u}})}^H) \wb_i\bigg|^2 + \sigma_k^2\bigg)}\Bigg] -\nonumber \\&\quad
\log_2 \Bigg[\frac  {\bigg(\displaystyle\sum_{i = 1}^K \bigg|(\displaystyle\sum_{l=1} ^ L  {(\gb_{l}^{\text{IRS-e}})} ^H \Thetab_l \Fb_l + {(\gb^{\text{BS-e}})}^H) \wb_i\bigg|^2 + \sigma_e^2\bigg)} {\bigg(\displaystyle\sum_{i \neq k}\bigg|(\displaystyle\sum_{l=1} ^ L  {(\gb_{l}^{\text{IRS-e}})} ^H \Thetab_l \Fb_l + {(\gb^{\text{BS-e}})}^H) \wb_i\bigg|^2+\sigma_e^2\bigg)} \Bigg].
\end{align}
Similar to $\mathcal{P}1$, this sub-problem is also non-convex because the objective function is non-concave. But in contrast to the previous sub-problem, there is also a non-convex constraint \eqref{eq.constraint.P2_1}. To deal with it, we first propose an equivalent form of this problem by introducing Proposition 2 below. 
		
\noindent{\bf Proposition 2.} \emph{The optimization problem $\mathcal{P}2$ can be reformulated as the following optimization problem.}	
	\begin{subequations}
		\label{P2_Eq}
		\begin{align}
		\label{P2_obj}
&\max_{\boldsymbol{\alpha},m_k,n_k,u,z_k,b,\epsilon_{l,n}} \: \: (-u+b+P_e\sum_{l,n}\epsilon_{l,n}) \quad\quad\quad\quad \quad\quad \\ 		\label{P2_Eq_c1}
&	\qquad\text{s.t.}  \: \: \: \eqref{eq.constraint.P2}, \eqref{eq.constraint.P2_1}, \eqref{eq.constraint.P2_2} \\ \label{P2_Eq_c2}
&\qquad\: m_k-n_k+z_k \geq b , \: \forall k, \\ \label{P2_Eq_c3}
& \qquad\: \sum_{i = 1} ^K \big| \boldsymbol{\alpha}^T \hb_{k,i}' + \hb_{k,i}''\big|^2 + \sigma^2_k \ge 2^{m_k} , \: \forall k, \\ \label{P2_Eq_c4}
&\qquad\:	\sum_{i \neq k}  \big| \boldsymbol{\alpha}^T \hb_{k,i}'+ \hb_{k,i}''\big|^2 + \sigma^2_k \le 2^{n_k} , \: \forall k, \\ \label{P2_Eq_c5}
&\qquad \:	\sum_{i = 1} ^K \big|\boldsymbol{\alpha}^T  \gb_i' +\gb_i''\big|^2 + \sigma_e^2 \le 2^u, \quad\quad\quad \\\label{P2_Eq_c6}
&\qquad \:	\sum_{i \neq k}  \big| \boldsymbol{\alpha}^T  \gb_i' +\gb_i''\big|^2 + \sigma_e^2 \ge 2^{z_k} , \: \forall k, 
		\end{align}
	\end{subequations}
\emph{where $\hb_{k,i}'= \text{diag}(\hb_{k}^{\text{IRS-u}})^* \Fb \wb_i \in \mathbb{C}^{LN \times 1}, \hb_{k,i}''= {(\hb_{k}^{\text{BS-u}})}^H \wb_i \in \mathbb{C}^{LN \times 1}$, and $\gb_i'=  \text{diag}(\gb^{\text{IRS-e}})^* \Fb \wb_i \in \mathbb{C}^{LN \times 1}$, $\gb_i''= {(\gb^{\text{BS-e}})}^H \wb_i \in \mathbb{C}^{LN \times 1}$. Furthermore,	${(\hb_{k}^{\text{IRS-u}})}^H=\big[{(\hb_{1,k}^{\text{IRS-u}})}^H , ... , {(\hb_{L,k}^{\text{IRS-u}})}^H\big] \in \mathbb{C}^{1 \times LN}$,\\ ${(\gb^{\text{IRS-e}})}^H=\big[{(\gb_{1}^{\text{IRS-e}})}^H , ... , {(\gb_{L}^{\text{IRS-e}})}^H\big] \in \mathbb{C}^{1 \times LN}$, $\Fb=[\Fb_1^T, ... , \Fb_L^T]^T \in \mathbb{C}^{LN \times M}$, and $\boldsymbol{\alpha} = [\boldsymbol{\alpha}_1^T,...,\boldsymbol{\alpha}_L^T]^T \in\mathbb{C}^{LN \times 1}$.}
\begin{IEEEproof} 
	Assuming that all the IRSs have the same number of phase shifts, e.g., $N_l= N, \forall l$, by defining ${(\hb_{k}^{\text{IRS-u}})}^H=\big[{(\hb_{1,k}^{\text{IRS-u}})}^H , ... , {(\hb_{L,k}^{\text{IRS-u}})}^H\big] \in \mathbb{C}^{1 \times LN}$,  ${(\gb^{\text{IRS-e}})}^H=\big[{(\gb_{1}^{\text{IRS-e}})}^H , ... , {(\gb_{L}^{\text{IRS-e}})}^H\big] \in \mathbb{C}^{1 \times LN}$, $	\Fb=[\Fb_1^T, ... , \Fb_L^T]^T \in \mathbb{C}^{LN \times M}$, $\Thetab = \text{diag}(\Thetab_1, ... ,\Thetab_L)\in \mathbb{C}^{LN \times LN}$ and $\boldsymbol{\alpha} = [\boldsymbol{\alpha}_1^T,...,\boldsymbol{\alpha}_L^T] \in\mathbb{C}^{LN \times 1}$, a part of objective function in $\mathcal{P}2$ can be reformulated as
			\begin{align}
		\label{eq.objective.P2.taghir.yafte}
		&\nonumber \qquad \qquad \qquad \qquad \qquad R_{s,k}=\\
		&\log_2\bigg(\frac{\displaystyle\sum_{i=1}^K\bigg|\big(   {(\hb_{k}^{\text{IRS-u}})} ^H \Thetab \Fb + {(\hb_{k}^{\text{BS-u}})}^H\big) \wb_i\bigg|^2 + \sigma^2_k}  {\displaystyle\sum_{i \neq k} \bigg|\big({(\hb_k^{\text{IRS-u}})}^H \Thetab \Fb + {(\hb_{k}^{\text{BS-u}})}^H\big) \wb_i\bigg|^2 + \sigma^2_k} \bigg) - \nonumber \\
		&\qquad\quad\log_2 \bigg( \frac  {\sum_{ i = 1}^K \bigg|( {(\gb^{\text{IRS-e}})} ^H \Thetab \Fb + {(\gb^{\text{BS-e}})}^H) \wb_i\bigg|^2 + \sigma_e^2} {\displaystyle\sum_{i \neq k}\bigg|(  {(\gb^{\text{IRS-e}})} ^H \Thetab \Fb + {(\gb^{\text{BS-e}})}^H) \wb_i\bigg|^2+\sigma_e^2} \bigg). 
		\end{align}
  After that, using the fact that $\Thetab$ is a diagonal matrix, we have the following equations
    \begin{subequations}
  \begin{align}
  	&(\hb_{k}^{\text{IRS-u}}) ^H \Thetab = \boldsymbol{\alpha}^T  \text{diag}(\hb_{k}^{\text{IRS-u}})^*., \\
  	&\big({\gb^{\text{IRS-e}})} ^H \Thetab =   \boldsymbol{\alpha}^T  \text{diag}(\gb^{\text{IRS-e}})^* .
  \end{align}
\end{subequations}
 where $\boldsymbol{\alpha} = [\boldsymbol{\alpha}_1,...,\boldsymbol{\alpha}_L] \in\mathbb{C}^{LN \times 1}$.
 Then, by defining $2^{m_k}=\sum_{i = 1} ^K \big| \big(   \boldsymbol{\alpha}^T\text{diag}(\hb_{k}^{\text{IRS-u}})^* \Fb + {(\hb_{k}^{\text{BS-u}})}^H\big) \wb_i\big|^2 + \sigma^2_k, \: \forall k $, 
  $2^{n_k} = \sum_{i \neq k}  \big|\big(\boldsymbol{\alpha}^T  \text{diag}(\hb_{k}^{\text{IRS-u}})^*\Fb + {(\hb_{k}^{\text{BS-u}})}^H\big) \wb_i\big|^2 + \sigma^2_k, \: \forall k$,
  $2^u=\sum_{i = 1} ^K \big|\big(\boldsymbol{\alpha}^T  \text{diag}(\gb^{\text{IRS-e}})^* \Fb + {(\gb^{\text{BS-e}})}^H\big) \wb_i\big|^2 + \sigma_e^2$, and $2^{z_k}=\sum_{i \neq k}  \big|\big(  \boldsymbol{\alpha}^T\text{diag}(\gb^{\text{IRS-e}})^* \Fb + {(\gb^{\text{BS-e}})}^H\big) \wb_i\big|^2 + \sigma_e^2, \forall k$, as auxiliary variables, $\mathcal{P}2$ can be rewritten as
\begin{subequations}
  	\begin{align}
  		&\max_{\boldsymbol{\alpha},m_k,n_k,u,z_k,\epsilon_{l,n}} \: \:  \min_k \: (m_k - n_k -u + z_k)+P_e\sum_{l,n}\epsilon_{l,n} \quad  \\\label{eq24b}
  		&\text{s.t.}  \: \: \: \eqref{eq.constraint.P2}, \eqref{eq.constraint.P2_1}, \eqref{eq.constraint.P2_2}, \\
  		& \: \sum_{i = 1} ^K \big| \big(\boldsymbol{\alpha}^T  \text{diag}(\hb_{k}^{\text{IRS-u}})^* \Fb + {(\hb_{k}^{\text{BS-u}})}^H\big) \wb_i\big|^2 + \sigma_k^2 \ge 2^{m_k} , \: \forall k, \\ 
  		 & \:	\sum_{i \neq k}  \big| \big(\boldsymbol{\alpha}^T  \text{diag}(\hb_{k}^{\text{IRS-u}})^* \Fb + {(\hb_{k}^{\text{BS-u}})}^H\big) \wb_i\big|^2 + \sigma_k^2 \le 2^{n_k} , \: \forall k, \\ 
  		 & \:	\sum_{i = 1} ^K \big| \big(\boldsymbol{\alpha}^T  \text{diag}(\gb^{\text{IRS-e}})^*  \Fb + {(\gb^{\text{BS-e}})}^H\big) \wb_i\big|^2 + \sigma_e^2 \le 2^u, \quad\quad\quad \\
  		 & \:\sum_{i \neq k}  \big|\big(\boldsymbol{\alpha}^T  \text{diag}(\gb^{\text{IRS-e}})^*  \Fb + {(\gb^{\text{BS-e}})}^H\big) \wb_i\big|^2 + \sigma_e^2 \ge 2^{z_k} , \: \forall k. 
  	\end{align}
  \end{subequations}
Finally, by introducing auxiliary variable $b=\underset{k}{\text{min}}(m_k-n_k+z_k)$ and also defining $\hb_{k,i}'= \text{diag}(\hb_{k}^{\text{IRS-u}})^* \Fb \wb_i$, $\hb_{k,i}''= {(\hb_{k}^{\text{BS-u}})}^H \wb_i$, and $\gb_i'=  \text{diag}(\gb^{\text{IRS-e}})^*\Fb \wb_i, \: \: \gb_i''= {(\gb^{\text{BS-e}})}^H\wb_i$, \eqref{P2_Eq_c2} constitutes, optimization problem $\mathcal{P}2$ is reformulated as \eqref{P2_Eq} and the proof is completed.
\end{IEEEproof}
Objective function in \eqref{P2_obj} is affine, thus it is concave. \eqref{P2_Eq_c2}, \eqref{eq.constraint.P2}, \eqref{eq.constraint.P2_2} are also convex constraints. The other constraints involving \eqref{eq.constraint.P2_1}, \eqref{P2_Eq_c3}, \eqref{P2_Eq_c4}, \eqref{P2_Eq_c5} and \eqref{P2_Eq_c6} are non-convex ones. All of these constraints have the property of convexity for both sides of the inequalities. Therefore, similar to the Lemma 2, in the following, we introduce Lemma 3 to combat non-convexity of these constraints by employing the SCA method.

\noindent{\bf Lemma 3}.
\eqref{eq.constraint.P2_1}, \eqref{P2_Eq_c3}, \eqref{P2_Eq_c4},\eqref{P2_Eq_c5}, and \eqref{P2_Eq_c6} can be approximated as the following convex constraints using the SCA approach.
\begin{subequations}
	\begin{align}
	\label{relaxed2}
	&  |\bar{\alpha}_{l,n}|^2+2\bar{\alpha}_{l,n}^{\Re}(\alpha_{l,n}^{\Re}-\bar{\alpha}_{l,n}^{\Re})+\nonumber\\& \qquad \qquad \qquad \qquad \qquad 2\bar{\alpha}_{l,n}^{\Im}(\alpha_{l,n}^{\Im}-\bar{\alpha}_{l,n}^{\Im}) \geq 1-\epsilon_{l,n}, \forall l,n \\\label{sca2_c1} 
	&  \sum_{i = 1} ^K \big| \bar{\boldsymbol{\alpha}}^T \hb_{k,i}' + \hb_{k,i}''\big|^2 +2 \sum_{i=1}^K \big[({\boldsymbol{\chi}}_{i,k}^\Re)^T(\boldsymbol{\alpha}^{\Re} - \bar{\boldsymbol{\alpha}}^{\Re})+\nonumber\\&\qquad\qquad\qquad\qquad\quad({\boldsymbol{\chi}}_{i,k}^\Im)^T(\boldsymbol{\alpha}^{\Im} - \bar{\boldsymbol{\alpha}}^{\Im})\big] +\sigma^2_k \ge 2^{m_k} , \: \forall k, \\ \label{sca2_c2}
	&	\sum_{i \neq k}  \big| \boldsymbol{\alpha}^T \hb_{k,i}'+ \hb_{k,i}''\big|^2 + \sigma^2_k \le 2^{\bar{n}_k} +\ln(2) 2^{\bar{n}_k} (n_k - \bar{n}_k), \: \forall k, \\ \label{sca2_c3}
	&	\sum_{i = 1} ^K \big|\boldsymbol{\alpha}^T  \gb_i' +\gb_i''\big|^2 + \sigma_e^2 \le 2^{\bar{u}} +\ln(2) 2^{\bar{u}} (u - \bar{u}), \\\label{sca2_c4}
	&	\sum_{i \neq k}  \big| \bar{\boldsymbol{\alpha}}^T  \gb_i' +\gb_i''\big|^2+2 \sum_{i \neq k} \big[({\boldsymbol{\lambda}}_{i}^\Re)^T(\boldsymbol{\alpha}^{\Re} - \bar{\boldsymbol{\alpha}}^{\Re})+\nonumber\\&\qquad\qquad\qquad\qquad\qquad({\boldsymbol{\lambda}}_{i}^\Im)^T(\boldsymbol{\alpha}^{\Im} - \bar{\boldsymbol{\alpha}}^{\Im})\big] + \sigma_e^2 \ge 2^{z_k} , \: \forall k,
	\end{align}
\end{subequations}
where $\boldsymbol{\alpha} = \boldsymbol{\alpha}^{\Re} + j\boldsymbol{\alpha}^{\Im}$, $\hb_{k,i}' = \hb_{k,i}'^{\Re} + j\hb_{k,i}'^{\Im}$, $\gb_i' = \gb_i'^{\Re} + j\gb_i'^{\Im}$. $\bar{\boldsymbol{\alpha}}^{\Re}$, $\bar{\boldsymbol{\alpha}}^{\Im}$, $\bar{n}_k$, and $\bar{u}$ are also the points where the first-order Taylor approximation are made. Moreover, ${\boldsymbol{\chi}}_{i,k}^\Re, {\boldsymbol{\chi}}_{i,k}^\Im, {\boldsymbol{\lambda}}_{i}^\Re,$ and ${\boldsymbol{\lambda}}_{i}^\Im \in \mathbb{C}^{M\times1}$ are denoted as
\begin{subequations}
	\begin{align}
	&\begin{bmatrix} 
	{\boldsymbol{\chi}}_{i,k}^\Re\\
	{\boldsymbol{\chi}}_{i,k}^\Im
	\end{bmatrix}
	=
	\begin{bmatrix} 
	\hb_{k,i}'^{\Re} & \hb_{k,i}'^{\Im} \\
	-\hb_{k,i}'^{\Im} & \hb_{k,i}'^{\Re} 
	\end{bmatrix}
	\begin{bmatrix} 
	\Re\big\{\bar{\boldsymbol{\alpha}}^T\hb_{k,i}'+\hb_{k,i}''\big\}  \\
	\Im\big\{\bar{\boldsymbol{\alpha}}^T\hb_{k,i}'+\hb_{k,i}'\big\}  
	\end{bmatrix} \\
	&\begin{bmatrix} 
	{\boldsymbol{\lambda}}_{i}^\Re\\
	{\boldsymbol{\lambda}}_{i}^\Im
	\end{bmatrix}
	=
	\begin{bmatrix} 
	\gb_{i}'^{\Re} & \gb_{i}'^{\Im} \\
	-\gb_{i}'^{\Im} & \gb_{i}'^{\Re} 
	\end{bmatrix}
	\begin{bmatrix} 
	\Re\big\{\bar{\boldsymbol{\alpha}}^T\gb_{i}'+\gb_{i}''\big\}  \\
	\Im\big\{\bar{\boldsymbol{\alpha}}^T\gb_{i}'+\gb_{i}'\big\}  
	\end{bmatrix}
	\end{align}
\end{subequations}
\begin{IEEEproof} 
	Due to the convexity of $2^{n_k}$ and $2^u$ in right-hand side of \eqref{P2_Eq_c4} and \eqref{P2_Eq_c5} their first-order Taylor expansion around $\bar{n}_k$ and $\bar{u}$ creates themselves lower bounds and leads to the convex-approximated constraints \eqref{sca2_c2} and \eqref{sca2_c3}, respectively. In contrast to \eqref{P2_Eq_c4} and \eqref{P2_Eq_c5}, in \eqref{eq.constraint.P2_1}, \eqref{P2_Eq_c3} and \eqref{P2_Eq_c6}, although greater sides of the inequalities are convex functions, they are not differentiable w.r.t. complex variable $\alpha_{l,n}$ and vector $\boldsymbol{\alpha}$. To tackle this problem, we break down $\alpha_{l,n}$ and $\boldsymbol{\alpha}$ into two distinctive real and imaginary parts as ${\alpha_{l,n}} = \alpha_{l,n}^{\Re} + j\alpha_{l,n}^{\Im}$, and $\boldsymbol{\alpha} = \boldsymbol{\alpha}^{\Re} + j\boldsymbol{\alpha}^{\Im}$, respectively. Then, according to Lemma 1, we perform first-order Taylor approximation around ${\bar{\alpha}_{l,n}} = \bar{\alpha}_{l,n}^{\Re} + j\bar{\alpha}_{l,n}^{\Im}$, and $\bar{\boldsymbol{\alpha}} = \bar{\boldsymbol{\alpha}}^{\Re} + j\bar{\boldsymbol{\alpha}}^{\Im}$ and approximately reformulate \eqref{eq.constraint.P2_1}, \eqref{P2_Eq_c3}, and \eqref{P2_Eq_c6} as convex constraints \eqref{relaxed2}, \eqref{sca2_c1} and \eqref{sca2_c4}.
\end{IEEEproof}	

Employing Lemma 3 as well as SCA method, non-convex problem \eqref{P2} can be approximated with a smaller feasible set than the original problem as the following iterative convex one.
	\begin{subequations}
	\label{P2_final}
	\begin{align}
	&\mathcal{P}2.1: \max_{\boldsymbol{\alpha},m_k,n_k,u,z_k,b,\epsilon_{l,n}} \: \: (-u+b+P_e\sum_{l = 1} ^{NL}\epsilon_{l,n}) \\ 	
	& \qquad\text{s.t.}  \: \:\: \:  |\alpha_{l,n}|^2 \leq 1 , \: \:   \:  \forall l,n, \\ 
	&\qquad\: |\bar{\alpha}_{l,n}^{(t_2)}|^2+2\bar{\alpha}_{l,n}^{\Re(t_2)}(\alpha_{l,n}^{\Re}-\bar{\alpha}_{l,n}^{\Re(t_2)})+\nonumber\\&\qquad\qquad\qquad\qquad 2\bar{\alpha}_{l,n}^{\Im(t_2)}(\alpha_{l,n}^{\Im}-\bar{\alpha}_{l,n}^{\Im(t_2)}) \geq 1-\epsilon_{l,n} , \: \:  \forall l,n \\ 
	& \qquad\: 0\leq\epsilon_{l,n}\leq 1,\forall l,n, \\ 
	&\qquad\: m_k-n_k+z_k \geq b , \: \forall k, \\
	& \sum_{i = 1} ^K \big|(\bar{\boldsymbol{\alpha}}^{(t_2)})^T \hb_{k,i}' + \hb_{k,i}''\big|^2 + 2 \sum_{i=1}^K \big[({\boldsymbol{\chi}}_{i,k}^{\Re(t_2)})^T(\boldsymbol{\alpha}^{\Re} - \bar{\boldsymbol{\alpha}}^{\Re(t_2)})+\nonumber\\&\qquad\qquad\qquad\quad({\boldsymbol{\chi}}_{i,k}^{\Im(t_2)})^T(\boldsymbol{\alpha}^{\Im} - \bar{\boldsymbol{\alpha}}^{\Im(t_2)})\big] +\sigma^2_k \ge 2^{m_k} , \: \forall k, \\ 
	&\qquad\:	\sum_{i \neq k}  \big| \boldsymbol{\alpha}^T \hb_{k,i}'+ \hb_{k,i}''\big|^2 + \sigma^2_k \le 2^{\bar{n}_k^{(t_2)}}+\nonumber\\& \qquad\qquad\qquad\qquad\qquad\qquad\qquad\ln(2) 2^{\bar{n}_k^{(t_2)}} (n_k - \bar{n}_k^{(t_2)}), \: \forall k, \\ 
	&\qquad \: \sum_{i = 1}^K \big|\boldsymbol{\alpha}^T  \gb_i' +\gb_i''\big|^2 + \sigma_e^2 \le 2^{\bar{u}^{(t_2)}} +\nonumber\\&\qquad\qquad\qquad\qquad\qquad\qquad\qquad\ln(2) 2^{\bar{u}^{(t_2)}} (u - \bar{u}^{(t_2)}), \quad\quad\quad \\
	& \:	\sum_{i \neq k} \big| (\bar{\boldsymbol{\alpha}}^{(t_2)})^T \gb_i' +\gb_i''\big|^2+2 \sum_{i \neq k} \big[({\boldsymbol{\lambda}}_{i}^{\Re(t_2)})^T(\boldsymbol{\alpha}^{\Re} - \bar{\boldsymbol{\alpha}}^{\Re(t_2)})+\nonumber\\&\qquad\qquad\qquad\qquad({\boldsymbol{\lambda}}_{i}^{\Im(t_2)})^T(\boldsymbol{\alpha}^{\Im} - \bar{\boldsymbol{\alpha}}^{\Im(t_2)})\big] + \sigma_e^2 \ge 2^{z_k}, \: \forall k,
	\end{align}
\end{subequations}	
where the optimization problem is relevant to the $t_2$-th iteration. \textbf{Algorithm 2} summarizes the proposed algorithm for solving $\mathcal{P}2$.     
     \begin{algorithm}[H]
     	\caption{Proposed algorithm for solving $\mathcal{P}2$.}
     	\begin{algorithmic}[1]   
     		\renewcommand{\algorithmicensure}{\textbf{Outputs:}}
     		\renewcommand{\algorithmicrequire}{\textbf{Initialization:}}	
     		\ENSURE  $ {\boldsymbol{\alpha}_1},...,{\boldsymbol{\alpha}_L}, \epsilon_{l,n}, \forall l,n.$		
     		\REQUIRE  ${t_{2,\text{max}}},$ ${\varepsilon_2},$ ${\wb_1},...,{\wb_K},$ $t_2=0,$ $\bar{n}_k^{(0)},$ $\bar{u}^{(0)},$ ${\bar{\boldsymbol{\alpha}}_1}^{(0)},...,{\bar{\boldsymbol{\alpha}}_K}^{(0)}$. 
     		\STATE \textbf{Repeat}
     		\STATE Solve the problem $\mathcal{P}2.1$ to obtain the optimal solution ${{\boldsymbol{\alpha}}_1}^{(t_2)},...,{{\boldsymbol{\alpha}}_K}^{(t_2)}$,
     		\STATE Set $ t_2=t_2+1 $,
     		\STATE Update $\bar{n}_k^{(t_2+1)}=n_k^{(t_2)}, \: \bar{u}^{(t_2+1)}=u^{(t_2)}$, ${\bar{\boldsymbol{\alpha}}_1}^{(t_2+1)}={{\boldsymbol{\alpha}}_1}^{(t_2)},...,{\bar{\boldsymbol{\alpha}}_K}^{(t_2+1)}={{\boldsymbol{\alpha}}_K}^{(t_2)}$ 
     		\STATE \textbf{Until} $ \left| {\frac{{\underset{k}{\text{min}} {R'}_{s,k}^{(t_2)} -\underset{k}{\text{min}} {R'}_{s,k}^{(t_2-1)}}}{{\underset{k}{\text{min}}{R'}_{s,k}^{(t_2-1)}}}} \right| \le {\varepsilon_2 } $ or $t_2 = t_{2,\text{max}}$.
     	\end{algorithmic}
     \end{algorithm}

     	In \textbf{Algorithm 2}, ${t_{2,\text{max}}}, {\varepsilon_2},{\wb_1,...,\wb_K}$ are the maximum number of iterations of the algorithm, maximum acceptable relative error of the objective function between two consecutive outputs, and $K$ BS beamforming vectors, respectively. Here, the objective function at the $t_2$-th iteration is defined as  $R_{s,k}^{'(t_2)} = R_{s,k}^{(t_2)} + P_e\sum_{l,n}\epsilon^{(t_2)}_{l,n}$.
     IRS phase profiles are outputs of this algorithm. 
     In the initialization step, we randomly generate  ${\bar{\boldsymbol{\alpha}}_1}^{(0)},...,{\bar{\boldsymbol{\alpha}}_L^{(0)}}$ with unit absolute value and uniform phase between 0 and $2\pi$. Considering the definition of auxiliary variables $n_k$ and $u$ in Proposition 2, $\bar{n}_k^{(0)}$ and $\bar{u}^{(0)}$ can be obtained as
     \begin{subequations}
		\begin{align}
		& \bar{n}_k^{(0)} = \log_2\bigg(\sum_{i \neq k}  \big| (\bar{\boldsymbol{\alpha}}^{(0)})^T \hb_{k,i}'+ \hb_{k,i}''\big|^2 + \sigma^2_k\bigg), \: \forall k,\\
		&\bar{u}^{(0)}=\log_2\bigg(\sum_{i = 1}^K \big|(\bar{\boldsymbol{\alpha}}^{(0)})^T  \gb_i' +\gb_i''\big|^2 + \sigma_e^2\bigg),
		\end{align}
		\end{subequations}      
 After solving $\mathcal{P}2.1$ at the $t_2$-th iteration, values for $\bar{n_k}^{(t_2+1)}$ and $\bar{u}^{(t_2+1)}$, ${\bar{\boldsymbol{\alpha}}_1}^{(t_2+1)},...,{\bar{\boldsymbol{\alpha}}_K}^{(t_2+1)}$ are achieved at step 4. After the update process, it is assured that the new values are in the feasible set of the optimization problem in the next iteration. This process continues until the convergence criterion is met. 
       \subsection{Overall Proposed Algorithm Based on BCD for Solving $\mathcal{P}'$}
       The overall proposed algorithm to solve $\mathcal{P'}$ (relaxed problem) is brought in details in \textbf{Algorithm 3}. In this algorithm, ${t_{\text{max}}}$ and ${\varepsilon}$ are the maximum number of iterations and maximum acceptable relative error of the objective function between two consecutive iterations, respectively. The objective function at the $t$-th iteration in this algorithm is defined as $R_{s,k}^{'(t)} = R_{s,k}^{(t)} + P_e\sum_{l,n}\epsilon^{(t)}_{l,n}$, which follows the same formulation as the previous algorithms.
       In the $t$-th iteration of \textbf{Algorithm 3}, using the obtained values of ${\boldsymbol{\alpha}_1}^{(t-1)},...,{\boldsymbol{\alpha}_L}^{(t-1)}$ in the previous iteration ($(t-1)$-th iteration) of the \textbf{Algorithm 2}, the sub-problem 1 is solved. After that, employing the results of \textbf{Algorithm 1}, \textbf{Algorithm 2} runs again. This process continues until the convergence criterion of the algorithm is satisfied. Convergence of all algorithms are also guaranteed which will be discussed in Section \ref{sec5}.
 
        \begin{algorithm}[H]
    	\caption{Proposed algorithm for solving $\mathcal{P}'$.}
    	\begin{algorithmic}[1]
    		\renewcommand{\algorithmicensure}{\textbf{Outputs:}}
    		\renewcommand{\algorithmicrequire}{\textbf{Initialization:}}	
    		\ENSURE  ${\wb_1},...,{\wb_K}$ and $ {\boldsymbol{\alpha}_1},...,{\boldsymbol{\alpha}_L}, \epsilon_{l,n},\forall l,n. $		
    		\REQUIRE ${t_{\text{max}}},$${\varepsilon}$, $t= 0$, ${\wb_1}^{(0)},...,{\wb_K}^{(0)},$ ${\boldsymbol{\alpha}_1}^{(0)},...,$ ${\boldsymbol{\alpha}_L}^{(0)},$ $\epsilon^{(0)}_{l,n}, \forall l,n$ 
    		\STATE \textbf{Repeat}
    		\STATE Given ${\boldsymbol{\alpha}_1}^{(t-1)},...,{\boldsymbol{\alpha}_L}^{(t-1)}$, solve $\mathcal{P}$1 based on \textbf{Algorithm 1} and obtain ${\wb_1}^{(t)},...,{\wb_K}^{(t)}$,
    			\STATE Given ${\wb_1}^{(t)},...,{\wb_K}^{(t)}$, solve $\mathcal{P}$2 based on \textbf{Algorithm 2} and obtain ${\boldsymbol{\alpha}_1}^{(t)},...,{\boldsymbol{\alpha}_L}^{(t)}$,
    		\STATE Set $ t=t+1 $, 
    		\STATE \textbf{Until} $ \left| {\frac{{\underset{k}{\text{min}} {R'}_{s,k}^{(t)} -\underset{k}{\text{min}} {R'}_{s,k}^{(t-1)}}}{{\underset{k}{\text{min}}{R'}_{s,k}^{(t-1)}}}} \right| \le {\varepsilon } $ or $t = t_{\text{max}}$.	
    	\end{algorithmic}
    \end{algorithm}
\subsection{Mapping Algorithm to Acquire Feasible IRSs Phase Shifts}
\label{mapping}
After solving $\mathcal{P}'$, because of continuous phase relaxation and also introduced penalizing concept, the obtained phase shifts are not in the feasible set in terms of both magnitude and phase. In the mapping algorithm, we map the obtained values of phase shifts to the feasible one according to the following criterion:
\begin{align}
\hat{\theta}_{l,n}=\underset{i\in\{0,...,N-1\}}{\text{argmin}}\big|\alpha_{l,n}-e^{j (i\Delta\theta)}\big|, \forall l,n
\end{align}
where $\Delta\theta = \frac{2\pi}{N}$ is the phase resolution of the IRSs. Here, we use Euclidean distance as the mapping metric because of its simplicity, low-complexity, and high performance even with a large number of phase-shifts.
\section {Performance Analysis}
\label{sec5}
Here, we analyze the convergence and computational complexity of our proposed algorithm.
\subsection{Convergence}
Before investigating the convergence property of the proposed algorithms, we introduce following theorem which will be used throughout the convergence analysis.
\begin{theorem}
	\label{t1_conv}
Consider a maximization problem with objective function $f_0(\xb)$ which is concave, bounded and differentiable, as well as convex and differentiable constraints $f_i(\xb) \ge 0,  \forall i\in\{1,...,N\}$, the SCA method based on first-order Taylor expansion of $f_i(\xb),~\forall i$ around $\bar{\xb}^{(t)}$ and updating rule $ {\bar \xb^{(t)}} = {\xb^{\ast(t - 1)}}$ converges in non-decreasing fashion to a Karush-Kuhn-Tucker (KKT) point of the original problem. 
\end{theorem}
\begin{proof}
	Refer to appendix \ref{Apen2}.	  
\end{proof}
According to the Theorem \eqref{t1_conv} and considering the fact that non-differentiable constraints \eqref{eqt3p} and \eqref{eqt5} have replaced by differentiable ones in Lemma2, \textbf{Algorithm 1} which is based on SCA method converges to a KKT point of $\mathcal{P}1$. Similar to \textbf{Algorithm 1}, convergence of \textbf{Algorithm 2} is also guaranteed to a KKT point of the original problem.

The following theorem reveals convergence of the overall proposed algorithm. 
 
 \begin{theorem}
 	\textbf{Algorithm 3} converges to a finite value in a non-decreasing fashion.
 \end{theorem}	    
 	 \begin{IEEEproof} 
 	 	Based on Theorem \ref{t1_conv}, both \textbf{Algorithm 1} and \textbf{Algorithm 2}, converge to their final points, in a non-decreasing fashion with the same objective function. The convergence property of \textbf{Algorithm 1} in the global iteration $t$ is as follows
 	 \begin{align}
 	 	\label{conv1}
 	 		&R'_s(\wb^{(t-1)(0)},{\boldsymbol{\alpha}^{(t-1)}},\boldsymbol{\epsilon}^{(t-1)}) \le\nonumber\\&\qquad\qquad\qquad\qquad	R'_s(\wb^{(t-1)(1)},{\boldsymbol{\alpha}^{(t-1)}},\boldsymbol{\epsilon}^{(t-1)})\le...\le\nonumber\\&\qquad\qquad\qquad\qquad\qquad\qquad 	R'_s(\wb^{(t-1)(t_1^{\text{conv}})},{\boldsymbol{\alpha}^{(t-1)}},\boldsymbol{\epsilon}^{(t-1)})
 	 	\end{align}
 	 	where $R'_s(\wb^{(t-1)(i)},{\boldsymbol{\alpha}^{(t-1)}},\boldsymbol{\epsilon}^{(t-1)}) = R_s(\wb^{(t-1)(i)},{\boldsymbol{\alpha}^{(t-1)}}) + P_e\sum_{l = 1} ^{NL}\epsilon^{(t-1)}_{l,n}$ denotes the total modified secrecy rate at the $t$-th iteration of \textbf{Algorithm 3} and the $i$-th iteration of \textbf{Algorithm 1}. $\wb$ is also considered as $\wb = [\wb_1^T,...,\wb_K^T]^T$ for notation simplification. Here is $t_1^{\text{conv}}$ is the iteration number which the convergence criterion of \textbf{Algorithm 1} is met.
 	 	
 	 	Similar to \textbf{Algorithm 1} this property for \textbf{Algorithm 2} is as follows
        \begin{align}
 	 	\label{conv2}
 	 	&R'_s(\wb^{(t-1)},{\boldsymbol{\alpha}^{(t-1)(0)}},\boldsymbol{\epsilon}^{(t-1)(0)})\le\nonumber\\&\qquad\qquad R'_s(\wb^{(t-1)},{\boldsymbol{\alpha}^{(t-1)(1)}},\boldsymbol{\epsilon}^{(t-1)(1)})\le ...\le\nonumber\\& \qquad\qquad\qquad\qquad\qquad R'_s(\wb^{(t-1)},{\boldsymbol{\alpha}^{(t-1)(t_2^{\text{conv}})}},\boldsymbol{\epsilon}^{(t-1)(t_2^{\text{conv}})})
 	 	\end{align}
 	 	where $R'_s(\wb^{(t-1)},{\boldsymbol{\alpha}^{(t-1)(i)}},{\boldsymbol{\epsilon}^{(t-1)(i)}})=R_s(\wb^{(t-1)},{\boldsymbol{\alpha}^{(t-1)(i)}}) + P_e\sum_{l = 1} ^{NL}\epsilon^{(t-1)(i)}_{l,n}$ denotes the total secrecy rate at the $t$-th iteration of \textbf{Algorithm 3} and the $i$-th iteration of \textbf{Algorithm 2}. Here, $t_2^{\text{conv}}$ is the iteration number which the convergence criterion of the \textbf{Algorithm 2} is met.
 	 	 
 	 	Following the convergence of each of the algorithms, the results are feasible in the first SCA iteration of the next algorithm, therefore, to move from \textbf{Algorithm 1} to \textbf{Algorithm 2} and vice versa, we have the following inequalities, respectively.
 	 	 \begin{align}
 	 	 \label{conv3}
 	 	 	&R'_s(\wb^{(t-1)(t_1^{\text{conv}})},{\boldsymbol{\alpha}^{(t-1)}},{\boldsymbol{\epsilon}^{(t-1)}})\le \nonumber\\&\qquad	R'_s(\wb^{(t-1)},{\boldsymbol{\alpha}^{(t-1)(1)}},{\boldsymbol{\epsilon}^{(t-1)(1)}})\le ...\le \nonumber \\& \qquad\qquad	R'_s(\wb^{(t-1)},{\boldsymbol{\alpha}^{(t-1)(t_2^{\text{conv}})}}, {\boldsymbol{\epsilon}^{(t-1)(t_2^{\text{conv}})}})\le \nonumber\\
 	 	 	& R'_s(\wb^{(t)(1)},{\boldsymbol{\alpha}^{(t)}},{\boldsymbol{\epsilon}^{(t)}}) \le...\le	R'_s(\wb^{(t)(t_1^{\text{conv}})},{\boldsymbol{\alpha}^{(t)}}, {\boldsymbol{\epsilon}^{(t)}}),
 	 	 \end{align}
 	 	 Summarizing \eqref{conv1}, \eqref{conv2} and \eqref{conv3} we have the following inequality.   	 	  
 	 	\begin{align}
 	 	&R'_s(\wb^{(t-1)},{\boldsymbol{\alpha}^{(t-1)}}, {\boldsymbol{\epsilon}^{(t-1)}}) \le R'_s(\wb^{(t)},{\boldsymbol{\alpha}^{(t-1)}},{\boldsymbol{\epsilon}^{(t-1)}}) \le \nonumber\\& \qquad \qquad\qquad\qquad\qquad\qquad\qquad\qquad\quad R'_s(\wb^{(t)},{\boldsymbol{\alpha}^{(t)}},{\boldsymbol{\epsilon}^{(t)}}),
 	 	\end{align}
 	 	The first inequality follows the fact that at the $t$-th iteration at the step 2 of \textbf{Algorithm 3}, for the given $\boldsymbol{\alpha}$ and $\boldsymbol{\epsilon}$ which is derived from the $(t-1)$-th iteration, $\wb$ is optimized. Thus, the objective function increases. In the second inequality, this process is repeated for the given values of $\wb$ and the objective function grows by optimizing $\boldsymbol{\alpha}$ and $\boldsymbol{\epsilon}$. Therefore, the secrecy rate improves at each iteration of \textbf{Algorithm 3}. Furthermore, limited resources such as power, number of antennas and IRSs, make the value of the secrecy rate be upper bounded. In other words, the algorithm converges to a finite value in a non-decreasing fashion and the proof is completed.
 	 \end{IEEEproof} 
 \subsection{Computational Complexity}
 	 According to appendix \ref{Apen3}, per-iteration computational complexity of the \textbf{Algorithm 1} and \textbf{Algorithm 2} are given by $\mathcal{O}(KM)^{3.5}$, and $\mathcal{O}(NL+K)^{3.5}$, respectively. Assuming it takes $T_1$ and $T_2$ iterations for \textbf{Algorithm 1} and \textbf{Algorithm 2} to converge to their final points based on an accuracy metric, the total complexity of these algorithms will be $\mathcal{O}\big(T_1(KM)\big)^{3.5}$  and $\mathcal{O}\big(T_2(NL+K)\big)^{3.5}$. Considering that the BCD method converges at $T_3$ iterations, computational complexity of \textbf{Algorithm 3} will be $\mathcal{O}\bigg(\big(T_1(KM)\big)^{3.5}+\big(T_2(NL+K)\big)^{3.5}\bigg)$. Mathematically, acquiring the exact number of iterations for convergence of these algorithms is very challenging. Thus, we achieve these values based on the simulations as we will discuss in the next section. By the way, the complexity of the mapping algorithm is $\mathcal{O}(N^2L)$ which is overlooked compared with the computational cost of \textbf{Algorithm 3}. Therefore, the total complexity equals to \textbf{Algorithm 3}-relevant one.

\section{Simulation Results}
\label{sec6}
\begin{table}[t]
	\renewcommand{\arraystretch}{1.3}
	\caption{Simulation Parameters}
	\centering
	\label{table_sim}
	\begin{tabular}{c|c|c}
		\hline
		\bfseries Parameter & \bfseries Values & \bfseries Statement \\
		\hline\hline
		BW & 100MHz & available bandwidth\\
		\hline
		$ B $ & $ 3 $ & number of paths in \eqref{eq1_channel}, \eqref{eq3_channel}, \eqref{eq7_channel}, and \eqref{eq8_channel}\\
		\hline
		$\mu_1$ & 61.4dB & constant path-loss term in \eqref{eq2_channel}\\
		\hline
		$\kappa_1$ & 2 & path-loss exponent in \eqref{eq2_channel}\\
		\hline
		$ \sigma_{\xi_1} $ & 5.8dB & shadowing standard deviation in \eqref{eq2_channel}\\
		\hline
		$\mu_2$ & 72dB & constant path-loss term in \eqref{eq5_channel}\\
		\hline
		$\kappa_2$ & 2.92 & path-loss exponent in \eqref{eq5_channel}\\
		\hline
		$ \sigma_{\xi_2} $ & 8.7dB & shadowing standard deviation in \eqref{eq5_channel}\\
		\hline
		$ \sigma_k^2 $ & -95dBm  &  thermal noise power at the user-side\\
		\hline
		$ \sigma^2_e $ & -95dBm  &  thermal noise power at the eavesdropper-side\\
		\hline
		$ \epsilon_1 $, $\epsilon_2$, $\epsilon_3$ & $10^{-3}$  &  acceptable relative error for the proposed algorithms\\
		\hline
	\end{tabular}
\end{table}	
In this section, simulation results are offered to evaluate the performance of the proposed algorithm. First, simulation setup and comparison benchmarks are introduced. After that, the convergence of the proposed algorithm is investigated. Finally, numerical results with detailed discussions are presented.
\subsection{Simulation Setup and Comparison Benchmarks}
\begin{figure}[t]
	\centering
	\includegraphics[width=9cm, height=8cm]{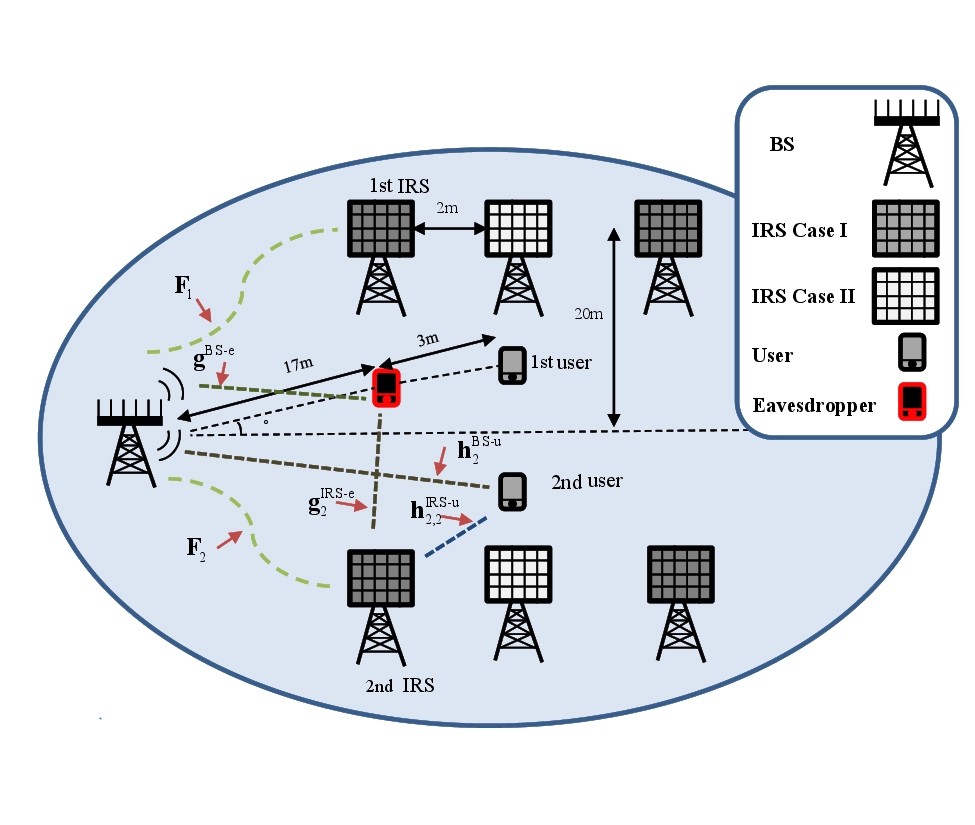}
	\caption{Our simulation setup consisting of a single BS, multiple IRSs, two users in presence of an eavesdropper.}
	\label{fig2}
\end{figure}

   Simulation parameters are listed in Table \ref{table_sim}. Fig. \ref{fig2} shows our simulation setup. The location of the BS is considered to be at the origin of the Cartesian coordinate system. As shown in Fig. \ref{fig2}, two cases of IRS numbers and distribution are considered where in the case I, $L=2$ and in the case II, $L=4$. In both cases, two users and an eavesdropper are symmetrically placed. It is assumed that the BS is equipped with a uniform linear array antennas, thus there is no possibility to perform vertical beamforming. Considering the worst case scenario, the eavesdropper is placed at a shorter distance than the legitimate users from the BS but in the same direction. 

The following benchmarks are selected for the performance comparison with our proposed scheme. 
	\begin{enumerate}
\item \textbf{Semi-definite Programming with Rank-one Relaxation (SDP-based algorithm):} The main benchmark to compare our results is  solution that is widely used in the literature \cite{jadid6, AH10, jadid7, jadid8, jadid1, jadid2}. In this method, using definition of some auxiliary variables, sub-problems $\mathcal{P}1$ and $\mathcal{P}2$, are re-written as SDP with a non-convex rank-one constraint for both sub-problems. Then, these constraints are relaxed and after solving the problem, a rank-one solution is approximated based on eigen-value decomposition (EVD) and Gaussian randomization method (EVD for $\mathcal{P}1$, and Gaussian randomization for $\mathcal{P}2$). This benchmark is called as SDP-based solution at the rest of the paper.
\item \textbf{MRT Benchmark:} The second benchmark for comparison is the MRT beamforming at the BS and optimizing IRSs phase shifts using \textbf{Algorithm 2}. In this benchmark, the complex conjugate of the channel between the BS and users is used for the active beamforming. In this method, without the need to prior information on the channel between BS and eavesdropper, and BS and IRSs reduces computational complexity by individually performing IRSs phase profiles and the BS beamforming vector optimization.  
\item \textbf{IRS-Free Benchmark:} This benchmark examines the network performance without the presence of any IRSs. By adding this benchmark, we seek to examine positive effect of IRS on the performance of our proposed network.
\end{enumerate}
\subsection{Rate of Convergence}
\begin{figure}[htp]
	\centering
	\includegraphics[width=9cm, height=8cm]{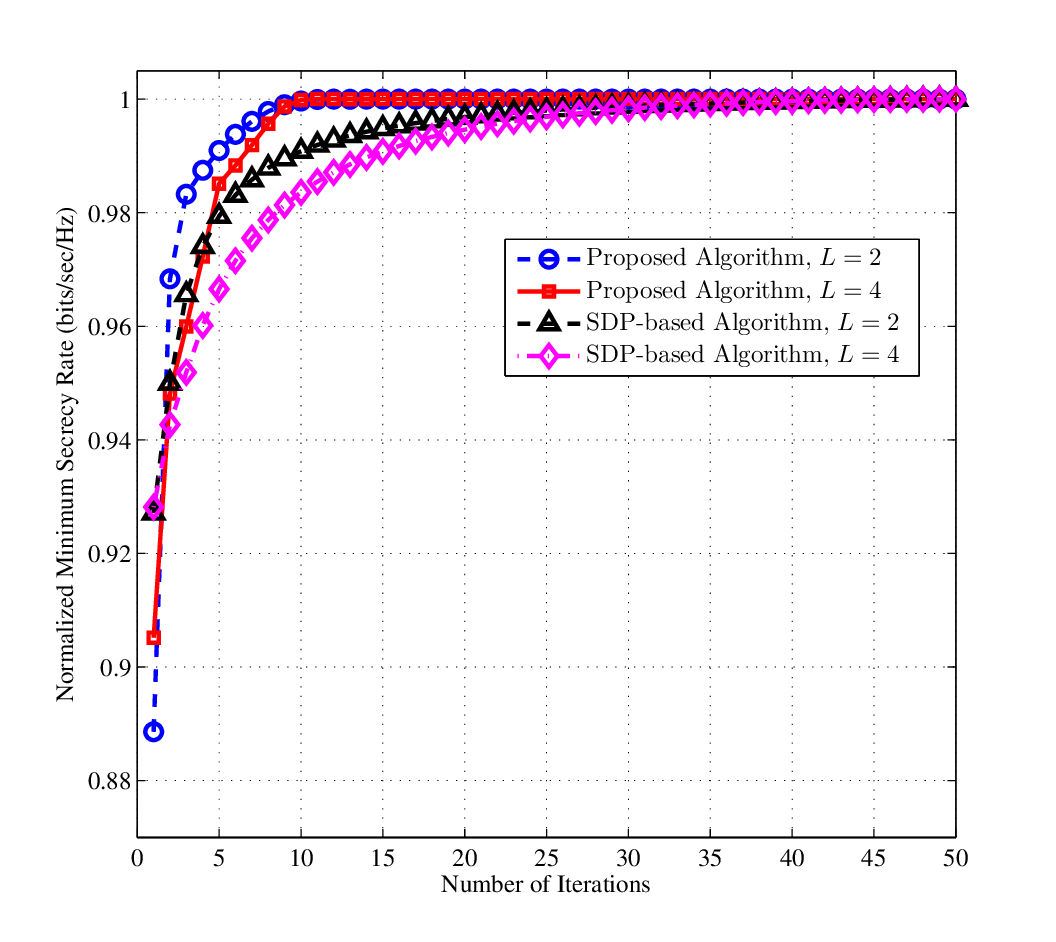}
	\caption{Convergence rate of the Algorithm 3.}
	\label{fig3}
\end{figure}
Fig. \ref{fig3} shows the sample convergence rate of \textbf{Algorithm 3} for both case I and case II compared with convergence of SDP-based solution as our main benchmark. In order to have a better informative comparison, we normalized the values of each curve by their own final points. As shown in this figure, our proposed algorithm outweighs SDP-based solution in terms of number of iterations to converge. In other words, for the same number of IRSs, our proposed algorithm converges faster than SDP-based one. This difference also grows for a larger number of IRSs. The main reason is that in SDP-based solution, semi-definite relaxation (SDR) leads to some approximated solutions for both sub-problems which reduces the relevant convergence rate. As growing number of IRSs leads to a larger number of approximated solutions due to employing Gaussian randomization and EVD method, lower convergence rate is achieved. In contrast, convergence rate of our proposed algorithm is almost insensitive to the number of IRSs.

\subsection{Numerical Results}
 All the convex optimization problems including our proposed one and benchmarks are solved using CVX optimization toolbox and the results are averaged over 1000 iterations for different channel realizations. 
\begin{figure}[t]
    	\centering
    	\includegraphics[width=9cm, height=8cm]{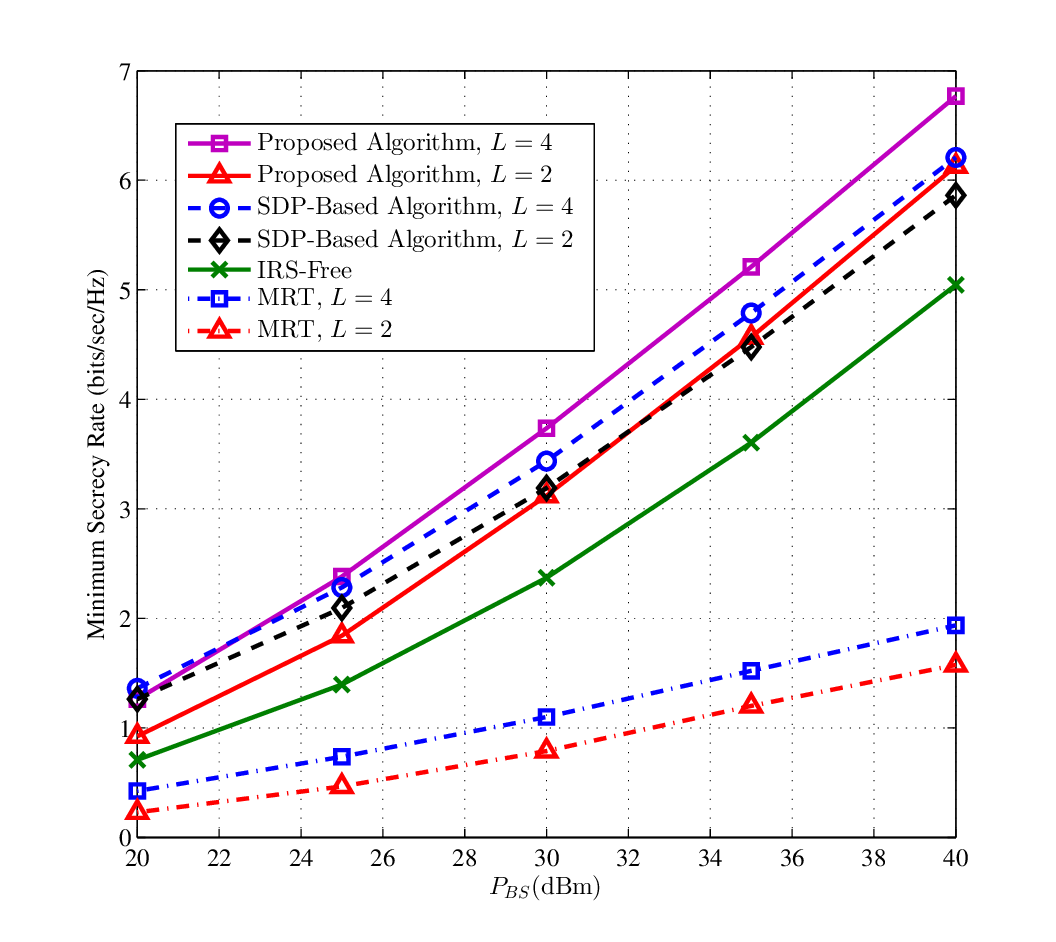}
    	\caption{Minimum secrecy rate versus the transmit power for two cases: $L=2$ and $L=4$, where $M=4$ and $N=16$.}
    	\label{fig4}
\end{figure}
\\
Fig. \ref{fig4} shows the minimum secrecy rate versus transmit power ($P_\text{BS}$) for Case I and Case II, where $M=4$ and $N=16$. It can be seen, in all criteria, as the $P_\text{BS}$ increases, higher values of the minimum secrecy rate are obtained. The reason lies in that with increasing BS power, the effect of beamforming on the received power of the user increases. On the other hand, the increase in BS's power improves the received power at the IRSs that leads to greater effect of phase shifts optimization on the minimum secrecy rate.
Our proposed algorithm with joint active and passive beamforming outperforms all the benchmarks. To be specific, our proposed algorithm has better performance than SDP-based especially at a larger number of IRSs where SDP-based suffers due to Gaussian randomization and EVD rank-one approximations. In addition, due to the lack of IRSs in the IRS-free benchmark, worse performance is obtained compared to the case with multiple IRSs.
In the MRT active beamforming vector design, as the effect of channels between BS and IRSs, and BS and eavesdropper are overlooked in the design process, worse performance than our proposed algorithm is obtained. Also, in this case the received power at the IRS-side is such low that their existence becomes less useful and results in lower performance than the IRS-Free case. More wastefulness of power consumption in the MRT beamforming than our proposed algorithm and IRS-Free case is another interesting point that can be deduced from increasing the gap among the curves at the high values of the transmission power.
Furthermore, this figure shows the importance of the number of IRSs in growth of the minimum secrecy rate, because more IRSs can increase the power level at the user-side and signal attenuation at eavesdropper. In general, Case II performs better than Case I due to a larger number of IRSs.
   \begin{figure}[t]
	\centering
	\includegraphics[width=9cm, height=8cm]{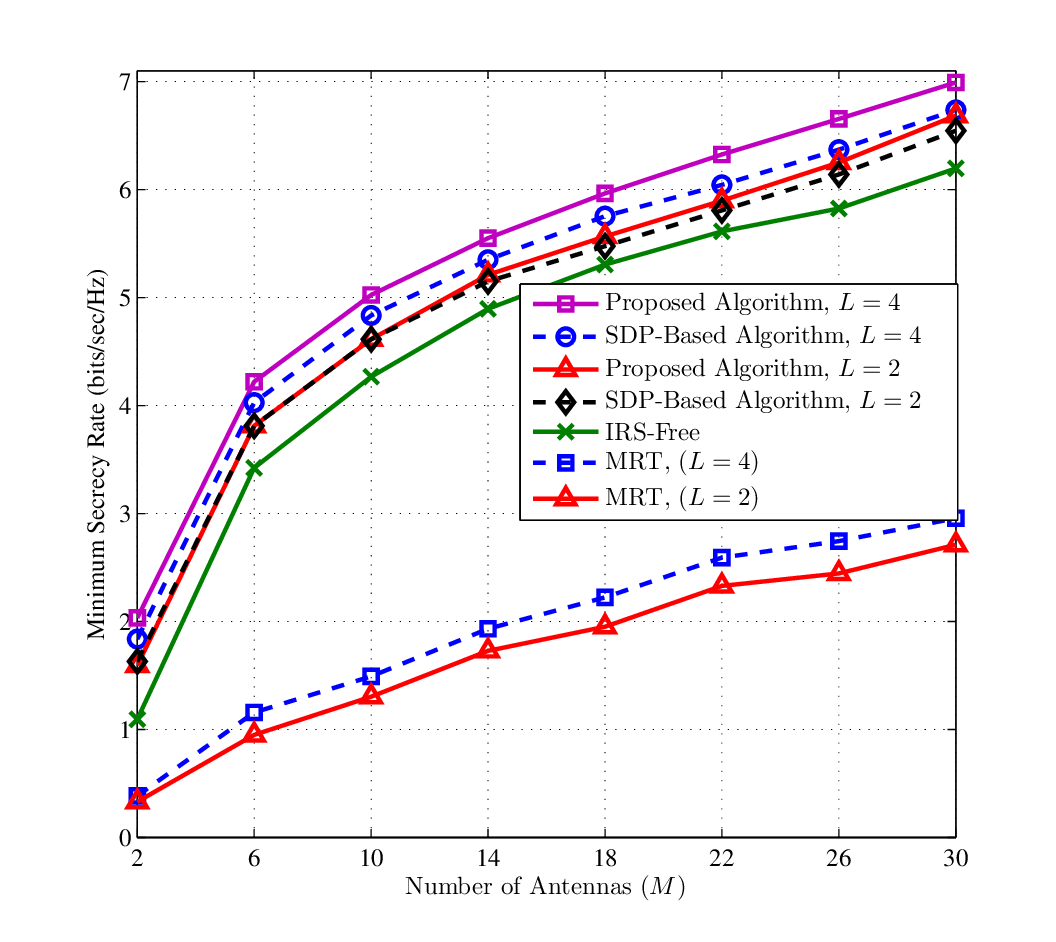}
	\caption{Minimum secrecy rate versus the number of BS antennas for two cases: $L=2$ and $L=4$.}
	\label{fig5}
\end{figure} 
\\
Fig. \ref{fig5} shows the minimum secrecy rate versus number of BS antennas $(M)$ for the proposed algorithm, SDP-based one, MRT and IRS-free benchmarks. In this simulation, the number of phase shifts and power of the BS are set to $N=$ 16 and $P_\text{BS}=$ 30dBm, respectively. As it is shown, there is an upward trend in minimum secrecy rate versus increasing number of antennas. In fact, as the number of antennas increases, the possibility of better beamforming at the BS increases, thus the secrecy rate improves. Similar to the previous figure, our proposed algorithm performs better than the all the benchmarks. Furthermore, there is an increase in gap between our proposed algorithm and SDP-based one. The reason falls into the fact that approximated solutions of SDR especially in sub-problem 1 performs worse with increasing number of antennas. Moreover, when the number of antennas increases, the gaps between the proposed algorithm and both IRS-free and MRT cases increases.
This is because the positive effect of beamforming in the proposed algorithm is more obvious than the others with increasing number of antennas.
   \begin{figure}[t]
	\centering
	\includegraphics[width=9cm, height=8cm]{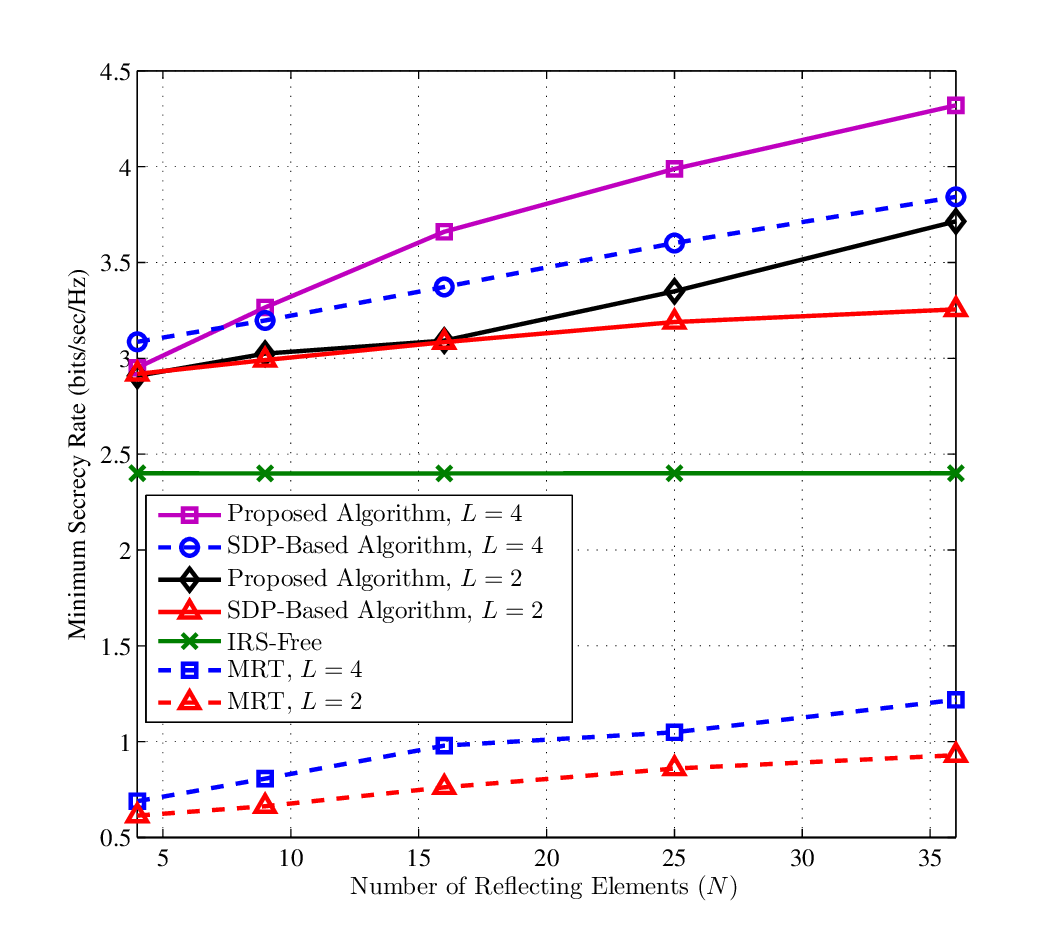}
	\caption{Minimum secrecy rate versus the number of IRS reflecting elements for two cases: $L=2$ and $L=4$.}
	\label{fig6}
\end{figure}
   \begin{figure}[t]
	\centering
	\includegraphics[width=9cm, height=8cm]{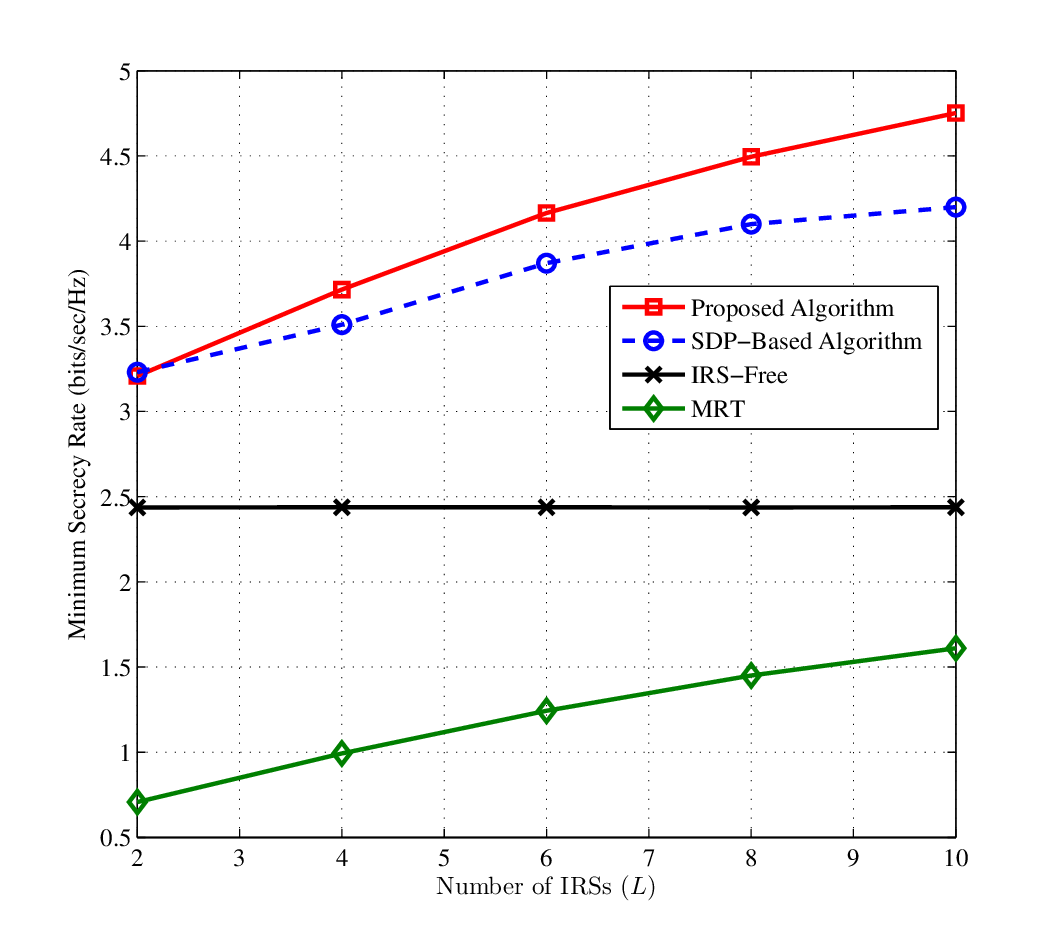}
	\caption{Minimum secrecy rate versus the number of IRSs.}
	\label{fig7}
\end{figure}
	
Fig. \ref{fig6} shows the minimum secrecy rate versus different number of IRS phase shifts $(N)$ for cases I and II. In this simulation the number of antennas and power at the BS are set as $M=4$ and $P_\text{BS}=30\text{dBm}$, respectively. Furthermore, in this simulation, the number of IRS phase shifts belongs to set $ \{4,9,16,25,36\}$. As the number of IRS elements increases, the minimum secrecy rate of the proposed algorithm, SDP-based and MRT improves.
Also, the gap between Case I and Case II tends to increase because as the number of IRS elements increases, the number of added IRS elements to Case II is twice more than that the number of elements added to Case I. An important point in this figure is that at a larger number of IRSs elements due to rate leakage in SDP-based algorithm because of rank-one approximations, our proposed solution with $L=2$ works approximately as equal as SDP-based one with $L=4$. In addition, as the number of IRS elements increases, the gap between the IRS-free and MRT case slightly increases.

In Fig. \ref{fig7}, the impact of number of IRSs on secrecy rate is elaborated. Here, we have considered $M=4$, $P_\text{BS}=30\text{dBm}$, and $N=16$. The placement of IRSs for $L=2$ and $L=4$ is similar to Fig. \ref{fig2}. For $L=6$, Case I and Case II are considered simultaneously. In addition, for $L=8$, two IRSs are symmetrically added to the relevant placement of $L=6$ and all of them are shifted to the right by $2m$. Similar to $L=8$, $L=10$ is also an extended case for placement of $L=8$ by adding two IRSs in a symmetric-manner. As expected, there is a growing trend in minimum secrecy rate by increasing number of IRSs in all cases but IRS-free. Interestingly, the gap between our proposed algorithm and SDP-based one goes up by increasing number of IRSs which is other evidence shows SDP suffers from a large number of optimization variables due to rank-one-approximation drawbacks.
\section{Conclusions}
\label{sec7}
We have investigated a multi-user multi-IRS system, where BS and IRSs cooperatively serve legitimate users and try to improve the network secrecy rate concurrently. Following the relaxation of discrete phase shifts to the continuous ones, joint active and passive beamforming for maximizing minimum secrecy rate has been solved via BCD and iteratively solving two relatively challenging-to-solve sub-problems. Each sub-problem alone has been solved iteratively via
SCA method. Penalty function method was also used to combat the non-convexity of phase shifts absolute values. Afterwards, a simple but effective mapping algorithm proposed to achieve feasible phase shifts. Then,
complexity and convergence behavior of the proposed algorithm has been characterized. Numerical results has corroborated the
improved performance of proposed joint design compared to conventional SDP-based solution with rank-one approximation, IRS-free and MRT benchmarks. 
\appendices
\section{Proof of Lemma 1}
\label{Apen1}
Assuming $\ab=\ab^{\Re}+j\ab^{\Im}$, $\xb=\xb^{\Re}+j\xb^{\Im}$, and $b=b^{\Re}+jb^{\Im}$, $y=|\ab^T\xb+b|^2$ can be written as the subsequent equation
\begin{align}
\label{lem1_eq}
&y=|\ab^T\xb+b|^2 = |(\ab^{\Re})^T\xb^{\Re}-(\ab^{\Im})^T\xb^{\Im}+b^{\Re} +\nonumber\\& \qquad\qquad\qquad\qquad\qquad j(\ab^{\Re})^T\xb^{\Im}+j(\ab^{\Im})^T\xb^{\Re}+jb^{\Im}|^2 = \nonumber\\
&\big|(\ab^{\Re})^T\xb^{\Re}-(\ab^{\Im})^T\xb^{\Im}+b^{\Re}\big|^2 + \big|(\ab^{\Re})^T\xb^{\Im}+(\ab^{\Im})^T\xb^{\Re}+b^{\Im}\big|^2.
\end{align} 
Following the fact that the second equation in \eqref{lem1_eq} is differentiable w.r.t. $\xb^{\Re}$ and $\xb^{\Im}$, its Gradient equals to \eqref{lemma1_main_eq}, and the proof is completed.
\section{Proof of Theorem \ref{t1_conv}}
\label{Apen2}
Firstly, the following lemma is introduced to help with the proof.\\  
\noindent{\bf Lemma}.
For a convex function $f(\xb)$ the following equality and inequality hold
\begin{subequations}
	\begin{align}
	&\tilde f(\xb,\bar \xb) \le f(\xb),\\
	&\tilde f(\bar \xb,\bar \xb) = f(\bar \xb)
	\end{align}
\end{subequations}
where ${{\tilde f}}(\xb,\bar \xb_{})$ is the first-order Taylor approximation of $f(x)$ around $\bar{\xb}$.
\begin{proof}
Based on the first-order criterion for convex functions, $f(\xb)$ is a convex one if and only if $f(\xb) \ge f(\bar \xb) + \nabla f{(\bar \xb)^T}(\xb - \bar \xb) = \tilde f(\xb,\bar \xb),$ where $\bar{\xb}$ belongs to the $f(\xb)$ domain. This inequality will be active if $\xb=\bar{\xb}$ and $\tilde f(\bar \xb,\bar \xb) = f(\bar \xb)$.	  
\end{proof}
Assuming that the successive convex approximated problem at the $t$-th iteration is in the following form
\begin{subequations}
	\begin{align}
	&{ \max_{\xb} }\quad {f_0}(\xb) \\
	&\mbox{s.t.} \quad{{\tilde f}_i}(\xb,\bar \xb_{}^{(t)}) \ge 0,\forall i \in {1,...,N}
	\end{align}
\end{subequations}
where ${{\tilde f}_i}(\xb,\bar \xb_{}^{(t)})$ is the first-order Taylor approximation of ${f_i}(\xb)$. According to the proposed Lemma, the left-hand sides of the constraints have the following upper bounds.
\begin{equation}
f_i(\xb) \ge {{\tilde f}_i}(\xb,\bar \xb_{}^{(t)}), \forall {i}
\end{equation}
We use the update rule $\bar \xb_{}^{(t)} = \xb_{}^{\ast(t - 1)}$, where $\xb_{}^{\ast(t - 1)}$ denotes the optimum value in the previous iteration. Since $\xb_{}^{\ast(t - 1)}$ achieves the upper bound of this inequality, it will be in the feasible set of the $t$-th iteration problem. Following this fact, the optimum value of the $t$-th iteration has the objective value at least equal to the previous iteration. Hence, the obtained objective values will be non-decreasing. Given that the objective function is bounded from above, the algorithm converges to a finite objective value. Note that here if the objective function is not a concave one, or in other words, if we fail to obtain approximated convex problem, we are not able to find the optimal solution in each iteration in a straightforward manner. Therefore, concavity of the objective function is sufficient condition for the objective function growing in a non-decreasing fashion. Now, we show that converging point satisfies the KKT for the original problem.

If we assume that the Slater condition is satisfied and objective and constraints functions are differentiable, the optimum value of $\xb$ given by $\xb^{\ast(t)}$ in the $t$-th iteration should satisfy the following KKT conditions \cite{boyd}
\begin{subequations} \label{sca_kkt}
	\begin{align}
	&\nabla {f_0}({\xb^{\ast(t)}}) + \sum\limits_{i = 1}^N {\lambda _i^*\nabla {\tilde{f}_i}({\xb^{\ast(t)}},{{\bar \xb}^{(t)}})}  = 0,\\
	&{\tilde{f}_i}({\xb^{\ast(t)}},{{\bar \xb}^{(t)}}) \ge 0,\qquad \lambda _i^* \ge 0,\\ 
	&\lambda _i^*{\tilde{f}_i}({\xb^{\ast(t)}},{{\bar \xb}^{(t)}}) = 0,
	\end{align}
\end{subequations}
where $\lambda _i^*$ are the optimum dual variables. As proven before, this SCA based algorithm converges for sufficiently large values of $t$, to $ \xb_{}^{\ast(t)} = \xb_{}^{\ast(t - 1)}$. Based on the updating rule we know that $\bar \xb_{}^{(t)} = \xb_{}^{\ast(t - 1)}$. Combining these we can write \eqref{sca_kkt} as
\begin{subequations}
	\begin{align}
	&\nabla {f_0}({\xb^{\ast(t)}}) + \sum\limits_{i = 1}^N {\lambda _i^*\nabla {f_i}({\xb^{\ast(t)}})}  = 0,\\
	&{{ f}_i}({\xb^{\ast(t)}}) \ge 0,\qquad \lambda _i^* \ge 0,\\
	&\lambda _i^*{{ f}_i}({\xb^{\ast(t)}}) = 0,
	\end{align}
\end{subequations}
This is the KKT conditions for the original problem. Thus the optimum value of the converged SCA sub-problems is also a KKT point of the original problem.
\section{Complexity Analysis }
\label{Apen3}
	One of the widely used proposed algorithms to solve non-linear convex optimization problems is interior point method \cite{polikgoogooli}. Assume the following non-linear convex optimization problem
\begin{align*}
&\underset{\xb}{\max}~~f(\xb)\\
&\mbox{s.t.}~g_j(\xb)\leq 0, \quad \forall j=1,2,\cdots,m\\
&\xb\geq \mathbf{0},
\end{align*}
where $\xb\in \mathbb{R}^{n}$. According to \cite[Theorem 4.2]{polikgoogooli}, it is required $\mathcal{O}(\sqrt{n+m})$ iterations with cost per iterations $\mathcal{O}((n+m)^3)$. Thus, the computational complexity will be  $\mathcal{O}((n+m)^{3.5})$ in total. According to this discussion, the overall cost per iteration of Algorithm 1 is $\mathcal{O}((5K+1+K(M+3)+2)^{3.5})$ given that $\mathcal{P}1.1$ has $K(M+3)+2$ variables and $5K+1$ constraints. Upon simplifying this expression, this will be  $\mathcal{O}((KM)^{3.5})$. For Algorithm 2, the complexity of each SCA step is $\mathcal{O}((2NL+3K+2+3NL+4K+1)^{3.5})$. This comes from the fact that $\mathcal{P}1.2$ has $2NL+3K+2$ variables and $3NL+4K+1$ constraints. Upon simplifying this expression, it will bring a complexity $\mathcal{O}((NL+K)^{3.5})$ for each step of Algorithm 2.
\bibliographystyle{IEEEtran}
\bibliography{IEEEabrv,library}

\end{document}